\documentclass[journal,twoside,web]{ieeecolor}
\pdfoutput=1
\usepackage{generic}
\usepackage{cite}
\usepackage{amsmath,amssymb,amsfonts}
\usepackage{algorithmic}
\usepackage{graphicx}
\usepackage{textcomp}
\usepackage{dblfloatfix}
\usepackage{caption}
\usepackage{amsthm}
\usepackage{arydshln}
\usepackage{tikz}
\usepackage[lined]{algorithm2e}

\newtheorem{theorem}{Theorem}[section]
\newtheorem{proposition}{Proposition}[section]
\newtheorem{remark}{Remark}[section]
\newtheorem{lemma}{Lemma}[section]

\def\BibTeX{{\rm B\kern-.05em{\sc i\kern-.025em b}\kern-.08em
		T\kern-.1667em\lower.7ex\hbox{E}\kern-.125emX}}
\begin{document}
	\title{Distributed Control of Descriptor Networks:\\A Convex Procedure for Augmented Sparsity}
	\author{Andrei Speril\u{a}$^\ddag$, Cristian Oar\u{a}$^\ddag$, Bogdan D. Ciubotaru$^\ddag$ and \c{S}erban Sab\u{a}u$^\sharp$
		\thanks{$\ddag$ Andrei Speril\u{a}, Cristian Oar\u{a} and Bogdan D. Ciubotaru  are with the Faculty of Automatic Control and Computers,
			``Politehnica" University of Bucharest, Sector 6, 060042 Romania \newline (e-mails: \{andrei.sperila, cristian.oara, bogdan.ciubotaru\}@upb.ro). }
		\thanks{$\sharp$ \c{S}erban Sab\u{a}u is with the Electrical and Computer Engineering Department, Stevens Institute of Technology, Hoboken, NJ 07030 USA \newline (e-mail:
			ssabau@stevens.edu).}
		\thanks{The first three authors were  supported by a grant of the Ministry of Research, Innovation and Digitization, CCCDI - UEFISCDI, project no. PN-III-P2-2.1-PED-2021-1626, within PNCDI III. The last author was supported by the NSF - CAREER award number 1653756.}\vspace{-5mm}
	}
	
	\maketitle
	\thispagestyle{empty}
	
	\begin{abstract}
		For networks of systems, with possibly improper transfer function matrices, we present a design framework which enables $\mathcal{H}_\infty$ control, while imposing sparsity constraints on the controller's coprime factors. We propose a convex and iterative optimization procedure with guaranteed convergence to obtain distributed controllers. By exploiting the robustness-oriented nature of our proposed approach, we provide the means to obtain sparse representations of our control laws that may not be directly supported by the network's nominal model.
	\end{abstract}
	
	\vspace{-3mm}
	\section{Introduction}\label{sec:intro}\vspace{-2mm}
	
	\subsection{Motivation}
	
	When faced with a distributed control problem, one notices an acute lack of dedicated numerical tools, if compared with the classical, centralized design context. Several computational methods, such as those proposed in \cite{Lavaei,reg,Sznaier,trian}, aim to exploit specialized techniques, in order to mitigate the numerical complexities inherent to distributed control.
	
	Notably, previous efforts \cite{Alav} have sought to enforce sparsity constraints directly upon a Finite Impulse Response (FIR) approximation of the Youla parameter, under certain restrictive assumptions, such as Quadratic Invariance (QI) and strong stabilizability (see \cite{QI_orig}). However, the technique proposed in Section 5 of \cite{Alav} cannot cope with enforcing sparsity patterns upon non-sparse \emph{affine expressions} of the Youla parameter.
	
	These issues were tackled in \cite{SLA}, with the introduction of the framework dubbed \emph{System-Level Synthesis} (SLS). Yet the focus on discrete-time systems meant that other architectures, such as the \emph{Network Realization Function} (NRF) representations discussed in \cite{plutonizare,NRF}, have been overshadowed by the FIR approximation methods from the SLS framework.
	
	\vspace{-4mm}
	
	
	\subsection{Paper structure and contributions}
	
	
	In this paper, we propose tractable techniques and numerical procedures for the NRF-based framework formalized in \cite{NRF}, which offers distributed control laws in \emph{both} continuous- and discrete-time, without needing to communicate any internal states, \emph{i.e.}, plant or controller states (see Section IV of \cite{NRF} for a comparison with the SLS framework), thus promoting scalable control laws for large-scale networks. In Section \ref{sec:prelim}, we cover a set of preliminary notions, with our paper's problem statement forming Section \ref{subsec:prob_st}. Our contributions are structured via the subsequent sections and may be listed as follows:
	\begin{enumerate}
		\item In Section \ref{subsec:mm_appr}, we show how to impose sparsity in the NRF formalism and how it reduces to a model-matching problem, that is solved via reliable procedures \cite{VSOLVE,VNULL,VCOVER};
		
		\item In Section \ref{subsec:rob_stab}, we extend the robust stabilization approach from \cite{glover1989robust} to the distributed, NRF-based setting;
		
		 \item In Section \ref{sec:algo}, we show how to particularize the convex and iterative procedure (with guaranteed convergence) from \cite{BMI2LMI} to obtain robust NRF-based implementations;
		
		\item In Section \ref{sec:examp}, we consider a generalization of the network in \cite{plutonizare} and we also show\footnote{All the implementations being compared in this paper are available at the following link: \texttt{https://github.com/AndreiSperila/CONPRAS}} how to employ our robustness-oriented approach to retrieve the same sparse control architecture as in \cite{plutonizare} for a more general case.
	\end{enumerate}
	Finally, Section \ref{sec:outro} contains a series of concluding remarks.

	\vspace{-1mm}
	\section{Preliminaries}\label{sec:prelim}
	
	\subsection{Nomenclature and definitions}\label{subsec:basic}

	Let $\mathbb{C}$, $\mathbb{C}^-$, $j\mathbb{R}$ and $\mathbb{B}$ denote the complex plane, the open left-half plane, the imaginary axis and the set $\{0,1\}$, respectively. Let $\mathbb{M}^{p\times m}$ stand for the set of all $p\times m$ matrices having entries in a set denoted $\mathbb{M}$. We also denote by $P\succ 0$ the fact that $P\in\mathbb{R}^{q\times q}$ is positive definite and by $\overline{\sigma}(Z)$ the maximum singular value of $Z\in\mathbb{C}^{p\times m}$. For any $M\in\mathbb{M}^{p\times m}$, $M^{\top}$ is its transpose. Let $\text{Ker}(M)$ denote the null space of $M\in\mathbb{M}^{p\times m}$ and let $\|Z\|_*$ denote the sum of the singular values belonging to $Z\in\mathbb{C}^{p\times m}$, which is termed the \emph{nuclear norm}. The operator $\otimes$ denotes the Kronecker product between any two matrices. We define the vectorization of $M\in\mathbb{M}^{p\times m}$ as $\text{vec}(M):=v\in\mathbb{M}^{pm\times 1}$, where $v_{i+(j-1)p}=M_{ij}$, along with the diagonalization of $M$ by $\text{diag}(M):=V\in\mathbb{M}^{pm\times pm}$, where $V_{ii}=\left(\text{vec}(M)\right)_i$, for $i\in1:pm$, and $V_{ij}=0$, $\forall i\neq j$. 
	
	For $M_i\in\mathbb{M}^{p_i\times m_i}$, with $i\in1:\ell$, and a natural number $g$, we define the block-diagonal concatenation operator by
	$
	\mathcal{D}(M_1,\dots,M_\ell)\hspace{-1mm}:=\hspace{-1mm}\footnotesize\begin{bmatrix}
	M_1&&\vspace{-2mm}\\
	&\ddots&\vspace{-2mm}\\
	&&M_\ell
	\end{bmatrix}\in\mathbb{M}^{\Sigma_{i=1}^{\ell}p_i\times\Sigma_{j=1}^{\ell}m_j}$ and the block-diagonal repetition of $M_i$ by $\mathcal{D}_g(M_i)\hspace{-1mm}:=\hspace{-1mm}\footnotesize\begin{bmatrix}
	M_i&&\vspace{-2mm}\\
	&\ddots&\vspace{-2mm}\\
	&&M_i
	\end{bmatrix}\in\mathbb{M}^{gp_i\times gm_i}
	$. For any $R\in\mathbb{M}^{q\times q}$, we denote its symmetric part by $\text{sym}(R):=\frac{1}{2}\big(R+R^{\top}\big)=\text{sym}\left(R^\top\right)$ and its diagonal part by
	$
	R^{\text{diag}}_{ij}:=\hspace{-1mm}\left\{\small\begin{aligned}
	R_{ij},\ i=j\\
	0,\ i\neq j\\
	\end{aligned}\right.,\ \forall i,j\in{1:q}
	$. 
	
	The matrix polynomial $A - s E$ is called a pencil, with square ones that have $\det(A - s E) \not \equiv 0$ being termed \emph{regular}. A regular pencil without finite generalized eigenvalues in $\mathbb{C}\backslash\mathbb{C}^-$ and without infinite generalized eigenvalues with partial multiplicities greater than $1$ (see \cite{gantmacher}) is called \emph{admissible}. Let $\Lambda(A-sE)$ be the collection of generalized eigenvalues (both finite and infinite) belonging to the regular pencil $A-sE$.

	In this paper, we will focus on systems described in frequency domain by {Transfer Function Matrices} (TFMs) of type
	$
	\left(\mathbf G (s)\right)_{ij}=\frac{a_{ij}(s)}{b_{ij}(s)},
	$
	with $a_{ij}(s)$ and $b_{ij}(s)$ polynomials with coefficients in $\mathbb{R}$, $i\in{1:p},\ j\in{1:m}$. We denote the set of all such TFMs with $m$ inputs and $p$ outputs by $\mathcal{R}^{p\times m}$, with $\mathcal{R}_{p}^{p\times m}$ being the subset of proper TFMs ($\deg a_{ij}\leq\deg b_{ij},\ \forall i\in1:p,\ j\in1:m$). Let $\mathcal{B}\in\mathbb{B}^{p\times m}$, which we use to express $\mathcal{S}_\mathcal{B}:=\{\mathbf{G}\in\mathcal{R}^{p\times m}\vert\mathcal{B}_{ij}=0\Rightarrow \mathbf{G}_{ij}\equiv0,\ \forall i\in{1:p},\ \forall j\in{1:m} \}.
	$ Note that
	$
	\mathbf{G}\in\mathcal{S}_{\mathcal{B}}\iff(I-\text{diag}(\mathcal{B}))\text{vec}(\mathbf{G})\equiv0.
	$
	We also define the restriction of $\mathcal{S}_\mathcal{B}$ to proper TFMs  $\widehat{\mathcal{S}}_\mathcal{B}:=\mathcal{S}_\mathcal{B}\cap\mathcal{R}_{p}^{p\times m}$. 
	
	A TFM without poles (see section 6.5.3 of \cite{Kai}) located in $\{\mathbb{C}\backslash\mathbb{C}^-\}\cup\{\infty\}$ is called \emph{stable}. Let $\mathcal{RH}_\infty$ denote the set of real-rational and stable TFMs, with the $\mathcal{H}_\infty$ norm of any $\mathbf G\in\mathcal{RH}_\infty$ being given by $\|\mathbf G\|_\infty := \sup_{s\in j{\mathbb{R}}} \overline{\sigma}\big(\mathbf G(s)\big)$. 
	
	The systems considered in this paper are usually represented in the time domain by differential and algebraic equations\vspace{-1mm}
	\begin{subequations}
		\begin{eqnarray}
		E\tfrac{\text{d}}{\text{d} t} x(t)&=&Ax(t)+Bu(t), \label{eq:l1} \\
		y(t) &=& Cx(t) + Du(t) \label{eq:l2}.
		\end{eqnarray}
	\end{subequations}
	The dimension of the regular pencil $A-sE$ and that of $x$, the vector which contains the realization's \emph{descriptor variables}, is called the {order} of the realization \eqref{eq:l1}-\eqref{eq:l2}. If its order is the smallest out of all others of its kind, a realization is called minimal (see section 2.4 of \cite{COAV}). Moreover, we have that\vspace{-1mm}
	\begin{equation}\label{eq:tfm}
	\mathbf G(s)=C(s E-A)^{-1}B + D=: \left[ \footnotesize\begin{array}{c|c} A - s E & B \\ \hline C & D \end{array} \right].\vspace{-1mm}
	\end{equation}
	
	Let the matrix $S_\infty$ span $\text{Ker}E$. A pair $(A-sE,B)$ or a realization \eqref{eq:tfm} for which $\left[ A - sE \, B \right]$ has full row rank $\forall s\in{\mathbb{C}}\backslash\mathbb{C}^-$ and $\left[ E \, AS_\infty \, B \right]$ has full row rank is called strongly stabilizable. By Theorem 1.1 in \cite{VPLACE}, strong stabilizability is equivalent to the existence of a matrix $F$, called an admissible feedback, such that the pencil $A+BF-sE$ is admissible. By duality, a pair $(C,A-sE)$ or realization \eqref{eq:tfm} is deemed strongly detectable if $(A^{\top}-sE^{\top},C^{\top})$ is strongly stabilizable.
	
	Let both $E_r$ and $D_r^\top D_r$ be invertible and consider\vspace{-1mm}
	\begin{equation}\label{eq:GCARE}
		\small\begin{array}{l}
			E_r^{\top}X_rA_r+A_r^{\top}X_rE_r+C_r^{\top}C_r-(E_r^{\top}X_rB_r+C_r^{\top}D_r)\times\\
			\hfill\times(D_r^{\top}D_r)^{-1}(B_r^{\top}X_rE_r+D_r^{\top}C_r)=0,
		\end{array}\vspace{-1mm}\hspace{-2mm}\normalsize
	\end{equation}
	the generalized continuous-time algebraic Riccati equation (GCARE, see \cite{RiccBig}). A symmetric solution $X_r$ of the GCARE is called stabilizing if  
	$
	F_r:=-(D_r^{\top}D_r)^{-1}(B_r^{\top}X_rE_r+D_r^{\top}C_r)
	$ is a stabilizing feedback, \emph{i.e.}, $\Lambda(A_r+B_rF_r-sE_r)\subset\mathbb{C}^-$.

	\vspace{-3mm}
	
	\subsection{Parametrization of all stabilizing controllers}\label{subsec:stab_fac}
	
	To obtain a tractable parametrization for NRF-based control laws, we employ the class of \emph{all} controllers which stabilize a network whose TFM $\mathbf{G}^n\in\mathcal{R}^{(p_u+p)\times (m_u+m)}$ is given by \vspace{-1mm}
	\begin{equation}\label{eq:part}
	\mathbf{G}^n=\left[\begin{array}{c:c}
	\mathbf{G}_{11}^n&\mathbf{G}_{12}^n\\\hdashline
	\mathbf{G}_{21}^n&\mathbf{G}_{22}^n
	\end{array}\right]=\left[\footnotesize\begin{array}{c|c:c}
	A-sE&B_1&B_2\\\hline
	C_1&D_{11}&D_{12}\\\hdashline
	C_2&D_{21}&D_{22}
	\end{array}\right],\normalsize\vspace{-1mm}
	\end{equation}
	where $A\in\mathbb{R}^{n\times n}$, $D_{11}\in\mathbb{R}^{p_u\times m_u}$, $D_{22}\in\mathbb{R}^{p\times m}$ and all other constant matrices have appropriate dimensions. 
	
	Under certain assumptions of strong stabilizability and detectability, the aforementioned class coincides with that of the controllers which render the closed-loop configuration from Fig. \ref{fig:cl} well-posed, \emph{i.e.}, $\text{det}(I-\mathbf{G}_{22}^n\mathbf{K})\not\equiv0$, and internally stable, \emph{i.e.}, all TFMs from $w_1$ and $w_2$ to $u_1$, $u_2$, $y_1$ and $y_2$ are stable. 
	We now state an extension of the Youla Parametrization, for a class of systems having possibly improper TFMs, by combining the notions from Sections 4.1 and 4.2 of \cite{takaba}.
	
	\begin{figure}[t]
		\centering
		\begin{tikzpicture}[scale=0.8, every node/.style={transform shape}]
			\node(n1)[circle,draw,minimum height=8]at(0,0){};
			\node(n2)[rounded corners=3, minimum height=1cm, minimum width=2cm,draw, thick, right of=n1, node distance=8em]{$ {\mathbf G}_{22}^n $};
			\node(n3)[rounded corners=3, minimum height=1cm, minimum width=2cm,draw, thick, below of=n2, node distance=4em]{${\mathbf K}$};
			\node(n4)[circle,draw,minimum height=8,right of=n3, node distance=8em]{$ $};
			\draw[-latex,line width=1.5pt](n1)++(-4em,0)--(n1) node[near start, above]{$w_1
				$} node[near end, below]{$+$};
			\draw[-latex,line width=1.5pt](n1)--(n2)node[ midway, above]{$u_1
				$};
			\draw[-latex,line width=1.5pt](n2)-|(n4) node[very near end, right]{$+$}node[ near start, above]{$y_1$};
			\draw[-latex,line width=1.5pt](n4)++(4em,0)--(n4) node[near start, above]{$w_2
				$} node[near end, below]{$+$};
			\draw[-latex,line width=1.5pt](n4)--(n3)node[ midway, above]{$\ u_2
				$};
			\draw[-latex,line width=1.5pt](n3)-|(n1) node[very near end, right]{$+$} node[ near start, above]{$y_2$};
		\end{tikzpicture}
		
		\caption{Closed-loop configuration}
		\label{fig:cl}\vspace{-5mm}
	\end{figure}
	
	\begin{theorem}\label{thm:Youla}
		Let $\mathbf{G}^n\in\mathcal{R}^{(p_u+p)\times (m_u+m)}$ be given as in \eqref{eq:part}, with $(A-sE,B_2)$ strongly stabilizable and $(C_2,A-sE)$ strongly detectable. Let $(\mathbf{N},\widetilde{\mathbf{N}},\mathbf{M},\widetilde{\mathbf{M}},\mathbf{X},\widetilde{\mathbf{X}},\mathbf{Y},\widetilde{\mathbf{Y}})$ be a doubly coprime factorization (DCF) of $\mathbf{G}_{22}^n=\mathbf{N}\mathbf{M}^{-1}=\widetilde{\mathbf{M}}^{-1}\widetilde{\mathbf{N}}$ over $\mathcal{RH}_{\infty}$, with all $8$ TFMs being stable and satisfying 
		\begin{equation}\label{eq:bez}
		\left[\small\begin{array}{rr}
		\widetilde{\mathbf Y}&-\widetilde{\mathbf X}\\
		-\widetilde{\mathbf N}&\widetilde{\mathbf M}
		\end{array}\right]\left[\small\begin{array}{cc}
		\mathbf M&\mathbf X\\
		\mathbf N&\mathbf Y
		\end{array}\right]=\left[\small\begin{array}{cc}
		I&0\\
		0&I
		\end{array}\right].\normalsize
		\end{equation} 
		Then, we have that:
		\begin{enumerate}
			\item[$\mathbf{(a)}$] A DCF over $\mathcal{RH}_\infty$ can be obtained, via \eqref{eq:part}, by
			\begin{subequations}
				\begin{align}
				\hspace{-2mm}\left[\small\begin{array}{c:c}
				\widetilde{\mathbf Y}&-\widetilde{\mathbf X}\\\hdashline
				-\widetilde{\mathbf N}&\widetilde{\mathbf M}
				\end{array}\right]\hspace{-1mm}:=&\hspace{-1mm}\left[\footnotesize\begin{array}{c|c:c}
				A_H-sE&-B_2-HD_{22}&H\\\hline
				F&I&0\\\hdashline
				C_2&-D_{22}&I
				\end{array}\right]\hspace{-1mm},\label{eq:dublu1}\\
				\left[\small\begin{array}{c:c}
				\mathbf M&\mathbf X\\\hdashline
				\mathbf N&\mathbf Y
				\end{array}\right]\hspace{-1mm}:=&\hspace{-1mm}\left[\footnotesize\begin{array}{c|c:c}
				A_F-sE&B_2&-H\\\hline
				F&I&0\\\hdashline
				C_2+D_{22}F&D_{22}&I
				\end{array}\right]\normalsize\hspace{-1mm},\label{eq:dublu2}
				\end{align}
			\end{subequations}
			with both of the pencils $
			A_H-sE:=A+HC_2-sE$ and $
			A_F-sE:=A+B_2F-sE$
			being admissible;
			
			\item[$\mathbf{(b)}$] The class of all stabilizing controllers is given by
			\begin{equation}\label{eq:K_frac}\small
			\mathbf K=(\mathbf X+\mathbf M\mathbf Q)(\mathbf Y+\mathbf N\mathbf Q)^{-1}\hspace{-1mm}=(\widetilde{\mathbf {Y}}+\mathbf Q\widetilde{\mathbf {N}})^{-1}(\widetilde{\mathbf {X}}+\mathbf Q\widetilde{\mathbf {M}}),\normalsize
			\end{equation}
			for all $\mathbf{Q}\in\mathcal{RH}_\infty^{m\times p}$ which ensure that $\det(\mathbf Y+\mathbf N\mathbf Q)\not\equiv0$ and $\det\big(\widetilde{\mathbf {Y}}+\mathbf Q\widetilde{\mathbf {N}}\big)\not\equiv0$;
			
			\item[$\mathbf{(c)}$] For a stabilizing $\mathbf K$ given by a DCF over $\mathcal{RH}_\infty$ of $\mathbf{G}_{22}^n$,
			\begin{equation*}\label{eq:H_CL}
			\mathbf{G}_{CL}=\mathcal{F}_\ell(\mathbf{G}^n\hspace{-1mm},\mathbf{K}):=\mathbf{G}_{11}^n+\mathbf{G}_{12}^n\mathbf{K}(I-\mathbf{G}_{22}^n\mathbf{K})^{-1}\mathbf{G}_{21}^n
			\end{equation*}
			is expressed affinely in terms of $\mathbf{Q}$ from \eqref{eq:K_frac} by the identity $\mathbf{G}_{CL}=\mathbf T_1+\mathbf T_2\mathbf
			Q\mathbf T_3$. Given a realization of the employed DCF, as in \eqref{eq:dublu1}-\eqref{eq:dublu2}, we have that
			\begin{subequations}
				\begin{align}
				\mathbf T_{1}:=&\ \mathbf{G}_{11}^n+\mathbf{G}_{12}^n\mathbf{X}\widetilde{\mathbf{M}}\mathbf{G}_{21}^n\nonumber\\
				=&\footnotesize\left[\begin{array}{cc|c}
				A_F-sE&-B_2F&B_1\\
				0&\hspace{-1mm}A_H-sE&B_1+HD_{21}\\\hline
				C_1+D_{12}F&-D_{12}F&D_{11}
				\end{array}\right],\label{eq:afin2}\\
				\mathbf T_{2}:=&\ \mathbf{G}_{12}^n\mathbf{M}=\footnotesize\left[\begin{array}{c|c}
				A_F-sE&B_2\\\hline C_1+D_{12}F&D_{12}
				\end{array}\right],\label{eq:afin3}\\
				\mathbf T_{3}:=&\ \widetilde{\mathbf{M}}\mathbf{G}_{21}^n=\footnotesize\left[\begin{array}{c|c}
				A_H-sE&B_1+HD_{21}\\\hline
				C_2&D_{21}
				\end{array}\right].\label{eq:afin4}
				\end{align}
			\end{subequations}
		\end{enumerate}
	\end{theorem}

	\begin{remark}\label{rem:dss2ss}
		The two admissible feedbacks $F$ and $H$ can always be chosen via the two step stabilization algorithm from \cite{VPLACE}. Since $A_H-sE$ and $A_F-sE$ are admissible, the TFMs from \eqref{eq:dublu1}-\eqref{eq:dublu2} and \eqref{eq:afin2}-\eqref{eq:afin4} are all stable, and thus proper. State-space realizations for these TFMs can be obtained via the residualization procedure mentioned in Section 3 of \cite{VNCF}.
	\end{remark}\vspace{-4mm}
	
	
	\section{Theoretical results}\label{sec:arch}
	
	
	\subsection{Problem statement}\label{subsec:prob_st}
	
	The results presented in this paper tackle the problem of obtaining \emph{sparse} and \emph{robustly stabilizing} control laws of type
	\begin{equation}\label{eq:NRF_explicit}
		\mathbf{u}_i=\textstyle\sum_{j=1}^{m}\mathbf{\Phi}_{ij}\mathbf{u}_j+\textstyle\sum_{k=1}^{p}\mathbf{\Gamma}_{ik}\mathbf{y}_k,\ \mathbf{\Phi}_{ii}\equiv0,\ \forall i\in1:m,
	\end{equation}
	discussed in \cite{NRF}. More specifically, we aim to impose $\mathbf{\Gamma}\in\widehat{\mathcal{S}}_\mathcal{X}$ and $\mathbf{\Phi}\in\widehat{\mathcal{S}}_\mathcal{Y}$, for some binary matrices $\mathcal{X}\in\mathbb{B}^{m\times p}$ and $\mathcal{Y}\in\mathbb{B}^{m\times m}$, with ${\mathcal{Y}}^{\text{diag}}=0$, and to have the control laws from \eqref{eq:NRF_explicit} stabilize all network models $\mathbf{G}_{\mathbf{\Delta}}\in\mathcal{C}_{\mathbf{G}}^\epsilon$, where the class $\mathcal{C}_{\mathbf{G}}^\epsilon$ is of the type discussed in \cite{gap}, owing to its generality.
	
	In the sequel, we show that this problem reduces to\vspace{-1mm}
	\begin{equation}\label{eq:prob}
		\left\|\widehat{\mathbf{T}}_1+\textstyle\sum_{i=1}^{q}\mathbf{x}_i\widehat{\mathbf{T}}_{2i}\right\|_\infty< 1,\ \mathbf{x}_i\in\mathcal{RH}_{\infty}^{1\times 1},\ i\in 1:q,\vspace{-1mm}
	\end{equation}
	with $\widehat{\mathbf{T}}_1,\widehat{\mathbf{T}}_{2i}\in\mathcal{RH}_\infty$, $\forall\ i\in 1:q$, being expressed in terms of \eqref{eq:tfm} and \eqref{eq:dublu1}-\eqref{eq:dublu2}. Finally, we particularize convex relaxation-based procedures \cite{BMI2LMI} available in literature to solve \eqref{eq:prob} and we compose $(\mathbf{\Phi},\mathbf{\Gamma})$ from the obtained $\mathbf{x}_i\in\mathcal{RH}_\infty^{1\times 1}$, $i\in 1:q$.
	

	\subsection{Parametrization of NRF-based control laws}\label{subsec:mm_appr}
	
	We show here how the problem of obtaining the \emph{sparse} and \emph{stabilizing} distributed control laws of type \eqref{eq:NRF_explicit} can be reduced to a readily solvable model-matching problem. As discussed in Section III of \cite{NRF}, this is primarily done by factorizing a stabilizing controller from the class expressed in Theorem \ref{thm:Youla} as
	$
	\mathbf{K}=(I-\mathbf{\Phi})^{-1}\mathbf{\Gamma}\label{eq:NRF1},
	$
	where $(\widetilde{\mathbf {Y}}+\mathbf Q\widetilde{\mathbf {N}})^{\text{diag}}$ and $(\widetilde{\mathbf {Y}}+\mathbf Q\widetilde{\mathbf {N}})$ have proper inverses and the NRF pair $(\mathbf{\Phi},\mathbf{\Gamma})$ is obtained as
	\begin{subequations}
		\begin{align}
			\mathbf{\Phi}:=&\ I-((\widetilde{\mathbf {Y}}+\mathbf Q\widetilde{\mathbf {N}})^{\text{diag}})^{-1}(\widetilde{\mathbf {Y}}+\mathbf Q\widetilde{\mathbf {N}})\in\mathcal{R}_{p}^{m\times m},\label{eq:NRF_def_1}\\
			\mathbf{\Gamma}:=&\ ((\widetilde{\mathbf {Y}}+\mathbf Q\widetilde{\mathbf {N}})^{\text{diag}})^{-1}(\widetilde{\mathbf {X}}+\mathbf Q\widetilde{\mathbf {M}})\in\mathcal{R}_{p}^{m\times p}.\label{eq:NRF_def_2}
		\end{align}
	\end{subequations}
	
	\begin{remark}\label{rem:implem}
		When the realization of $\mathbf{G}_{22}^n\in\mathcal{R}^{p\times m}$ (not necessarily proper) from \eqref{eq:part} is strongly stabilizable and detectable, the guarantees of closed-loop internal stability and of scalability showcased in Section III of \cite{NRF} for control laws of type \eqref{eq:NRF_explicit} will also hold. Thus, since all closed-loop transfers are stable and since the analogues of Lemmas 5.2 and 5.3 in \cite{zhou} (formulated for descriptor systems) are in effect, then the descriptor variables of both the plant and of the controller's NRF-based implementation (along with their output signals in closed-loop interconnection) will be bounded and will tend to $0$, when evolving freely from any finite initial conditions.
	\end{remark}
	
	With the stability guarantees of \eqref{eq:NRF_explicit} clarified in Remark \ref{rem:implem}, we now focus on imposing sparsity patterns on the $(\mathbf{\Phi},\mathbf{\Gamma})$ pair. 
	 The following result offers a characterization of the stable Youla parameters which, for a given DCF over $\mathcal{RH}_\infty$, produce the desired sparsity structure for the NRF pair in \eqref{eq:NRF_def_1}-\eqref{eq:NRF_def_2}.
	
	\begin{proposition}\label{prop:NRF}
		Let $\mathbf{G}\in\mathcal{R}^{p\times m}$ be given by a DCF over $\mathcal{RH}_\infty$ \eqref{eq:dublu1}-\eqref{eq:dublu2}, let $\mathcal{X}\in\mathbb{B}^{m\times p}$ and let $\widehat{\mathcal{Y}}\in\mathbb{B}^{m\times m}$, with $\widehat{\mathcal{Y}}^{\text{diag}}=I$. Define $F_{\mathcal{X}}:=I-\text{diag}(\mathcal{X})$ and $F_{\widehat{\mathcal{Y}}}:=I-\text{diag}(\widehat{\mathcal{Y}})$. If there exist $\mathbf{Q}_0\in\mathcal{RH}_\infty^{m\times p}$ and $\widehat{\mathbf{Q}}\in\mathcal{RH}_\infty^{m\times p}$ satisfying
		\begin{subequations}
			\begin{equation}\label{eq:NRF_EMM}
				\small\begin{bmatrix}
					F_{\mathcal{X}}(\widetilde{\mathbf{M}}^\top\otimes I)\\F_{\widehat{\mathcal{Y}}}(\widetilde{\mathbf{N}}^\top\otimes I)
				\end{bmatrix}\normalsize\text{vec}(\mathbf{Q}_0)+\small\begin{bmatrix}
					F_{\mathcal{X}}\text{vec}(\widetilde{\mathbf{X}})\\F_{\widehat{\mathcal{Y}}}\text{vec}(\widetilde{\mathbf{Y}})
				\end{bmatrix}\normalsize\equiv0,\vspace{-1mm}
			\end{equation}
			\begin{equation}\label{eq:null_NRF}
				\text{vec}(\widehat{\mathbf{Q}})\in\text{Ker}\small\begin{bmatrix}
					F_{\mathcal{X}}(\widetilde{\mathbf{M}}^\top\otimes I)\\F_{\widehat{\mathcal{Y}}}(\widetilde{\mathbf{N}}^\top\otimes I)
				\end{bmatrix}\normalsize,\vspace{-1mm}
			\end{equation}
			\begin{equation}\label{eq:mid_prop}
				\det\big((\widetilde{\mathbf {Y}}+(\mathbf{Q}_0+\widehat{\mathbf{Q}})\widetilde{\mathbf {N}})(\infty)\big)\neq0,\vspace{-2mm}
			\end{equation}
			\begin{equation}\label{eq:fin_prop}
				\det\big((\widetilde{\mathbf {Y}}+(\mathbf{Q}_0+\widehat{\mathbf{Q}})\widetilde{\mathbf {N}})^{\text{diag}}(\infty)\big)\neq0,\vspace{-2mm}
			\end{equation}
		\end{subequations}
		then the controller in \eqref{eq:K_frac}, formed via the employed DCF over $\mathcal{RH}_\infty$ of type \eqref{eq:dublu1}-\eqref{eq:dublu2} and via $\mathbf{Q}:=\mathbf{Q}_0+\widehat{\mathbf{Q}}$, admits an NRF implementation of \eqref{eq:NRF_def_1}-\eqref{eq:NRF_def_2} with $\mathbf{\Gamma}\in\widehat{\mathcal{S}}_\mathcal{X}$ and $\mathbf{\Phi}\in\widehat{\mathcal{S}}_{(\widehat{\mathcal{Y}}-I)}$.\vspace{-1mm}
	\end{proposition}
	\begin{proof}
		See the Appendix.\vspace{-1mm}
	\end{proof}
	
	\begin{remark}\label{rem:reuse_QI_orig}
		The equation \eqref{eq:NRF_EMM} can be solved for a stable $\text{vec}(\mathbf{Q}_0)$ as shown in \cite{VSOLVE}. Moreover, a least order solution can be obtained by employing the generalized minimum cover algorithm from \cite{VCOVER}. A benefit of this approach is that it computes a (stable) basis for $\text{Ker}\small\begin{bmatrix}
		F_{\mathcal{X}}(\widetilde{\mathbf{M}}^\top\otimes I)\\F_{\widehat{\mathcal{Y}}}(\widetilde{\mathbf{N}}^\top\otimes I)
		\end{bmatrix}$. Alternatively, a stable basis of least degree can be obtained as in \cite{VNULL}.\vspace{-1mm}
	\end{remark}
	
	\begin{remark}\label{rem:prop}
		Selecting a $\mathbf{Q}$ which ensures that $\det\big((\widetilde{\mathbf {Y}}+\mathbf Q\widetilde{\mathbf {N}})(\infty)\big)\neq0$, thus guaranteeing that the controller's TFM is well-posed, can be done numerically by using the fact that
		\vspace{-1mm}
		\begin{equation}\label{eq:BMI1}
			\hspace{-2mm}\begin{array}{l}
				\det\big(\widetilde{\mathbf {Y}}(\infty)+\mathbf{Q}(\infty)\widetilde{\mathbf {N}}(\infty)\big)\neq0\iff\\\hspace{-1mm}\iff\hspace{-1mm}\big(\widetilde{\mathbf {Y}}(\infty)+\mathbf{Q}(\infty)\widetilde{\mathbf {N}}(\infty)\big)^{\top}\big(\widetilde{\mathbf {Y}}(\infty)+\mathbf{Q}(\infty)\widetilde{\mathbf {N}}(\infty)\big)\succ 0.
			\end{array}\hspace{-2mm}\vspace{-1mm}
		\end{equation}
		To ensure $\det\big((\widetilde{\mathbf {Y}}+\mathbf Q\widetilde{\mathbf {N}})^{\text{diag}}(\infty)\big)\neq0$, we first denote by $e_i$ the $i^\text{th}$ vector of the canonical basis of $\mathbb{R}^{m\times 1}$ and impose that\vspace{-1mm}
		\begin{equation}\label{eq:BMI2}
			\hspace{-3mm}\begin{array}{l}
				e_i^{\top}(\widetilde{\mathbf {Y}}(\infty)+\mathbf Q(\infty)\widetilde{\mathbf {N}}(\infty))^{\top}e_i\times\\\quad\quad\times e_i^{\top}(\widetilde{\mathbf {Y}}(\infty)+\mathbf Q(\infty)\widetilde{\mathbf {N}}(\infty))e_i\succ 0,\ \forall i\in 1:m.
			\end{array}\hspace{-2mm}\vspace{-1mm}
		\end{equation}
		The bilinear matrix inequalities in \eqref{eq:BMI1}-\eqref{eq:BMI2} will be convexified and solved iteratively via the procedure given in Section \ref{sec:algo}.
	\end{remark}
	
	
	\vspace{-6mm}

	\subsection{Robust stabilization and augmented sparsity}\label{subsec:rob_stab}

	In this subsection, we show how to obtain a controller of type \eqref{eq:K_frac} whose NRF implementation \eqref{eq:NRF_explicit} stabilizes all network models $\mathbf{G}_{\mathbf{\Delta}}$ in a class $\mathcal{C}_{\mathbf{G}}^\epsilon$ and how this technique can be used to obtain a sparse control architecture. However, before this, we begin by defining the aforementioned class of TFMs.
	
	The class $\mathcal{C}_{\mathbf{G}}^\epsilon$, introduced in Section \ref{subsec:prob_st}, is expressed in terms of a stable right coprime factorization (RCF) of $\mathbf{G}=\widehat{\mathbf{N}}\widehat{\mathbf{M}}^{-1}\in\mathcal{R}^{p\times m}$, \emph{i.e.}, $\widehat{\mathbf{N}},\widehat{\mathbf{M}}\in\mathcal{RH}_\infty$ and $\exists\ \widetilde{\mathbf{X}},\widetilde{\mathbf{Y}}\in\mathcal{RH}_\infty$ so that $\widetilde{\mathbf{Y}}\widehat{\mathbf{M}}-\widetilde{\mathbf{X}}\widehat{\mathbf{N}}=I$, which is additionally normalized, \emph{i.e.}, $
	\widehat{\mathbf{N}}^{\top}(-s)\widehat{\mathbf{N}}(s)+\widehat{\mathbf{M}}^{\top}(-s)\widehat{\mathbf{M}}(s)=I
	$. With \emph{any} (see \cite{gap}) such stable normalized RCF (NRCF) and $\epsilon\in(0,1]$, we define\vspace{-2mm}
	\begin{equation}\label{eq:class}
	\small\begin{array}{llr}
	\hspace{-1mm}\mathcal{C}_{\mathbf{G}}^\epsilon\hspace{-1mm}&\hspace{-2mm}:=\hspace{-2mm}&\hspace{-1mm}\Big\{\big( \widehat{\mathbf N}+\mathbf{\Delta}_{\widehat{\mathbf{N}}}\big)\big(\widehat{\mathbf M}+\mathbf{\Delta}_{\widehat{\mathbf{M}}}\big)^{-1}\hspace{-1mm},\mathbf{\Delta}_{\widehat{\mathbf{N}}},\mathbf{\Delta}_{\widehat{\mathbf{M}}}\in   \mathcal{RH}_\infty,\\
	&&\det\big(\widehat{\mathbf M}+\mathbf{\Delta}_{\widehat{\mathbf{M}}}\big)\not\equiv 0,\ \left\|\begin{bmatrix}\mathbf{\Delta}_{\widehat{\mathbf{N}}}^\top& \mathbf{\Delta}_{\widehat{\mathbf{M}}}^\top\end{bmatrix}\hspace{-2mm}\phantom{.}^\top\right\|_\infty<\epsilon\Big\}.
	\end{array}\normalsize\vspace{-1mm}
	\end{equation}
	
	Clearly, in order to manipulate $\mathcal{C}_{\mathbf{G}}^\epsilon$, we must first obtain a stable NRCF of $\mathbf{G}$. While \eqref{eq:dublu2} readily provides a stable RCF of $\mathbf{G}$, a stable NRCF can be obtained via the following result.\vspace{-2mm}
	
	\begin{lemma}\label{lem:NRCF}
		Let $E_r$ be an invertible matrix and let also $\Lambda(A_r-sE_r)\subset\mathbb{C}^-$. Let the TFM\vspace{-1mm}
		\begin{equation}\label{eq:real_r}
		\begin{bmatrix}
		\mathbf{N}^\top&\mathbf{M}^\top
		\end{bmatrix}^\top=\left[\footnotesize\begin{array}{c|c}
		A_{r}-sE_{r}&B_{r}\\\hline C_{r}&D_{r}
		\end{array}\right]\in\mathcal{RH}_\infty^{(p+m)\times m}\vspace{-1mm}
		\end{equation}
		designate a stable RCF of $\mathbf{G}=\mathbf{N}\mathbf{M}^{-1}\in\mathcal{R}^{p\times m}$ and let $H_r\in\mathbb{R}^{m\times m}$ be invertible and satisfy $H_r^{\top}H_r=D_r^{\top}D_r$. Then:
		
		\begin{enumerate}
			\item[$\mathbf{(a)}$] The GCARE from \eqref{eq:GCARE} has a symmetric stabilizing solution, $X_r$, along with a stabilizing feedback, $F_r$;\smallskip
			
			\item[$\mathbf{(b)}$] For $\mathbf{G}_0:=\hspace{-1mm}\left[\footnotesize\begin{array}{c|c}
				A_r-sE_r&B_r\\\hline -H_rF_r&H_r
			\end{array}\right]$, we get that $\small\begin{bmatrix}
				\widehat{\mathbf{N}}^\top&\hspace{-3mm}\widehat{\mathbf{M}}^\top
			\end{bmatrix}\hspace{-2mm}\phantom{.}^\top\hspace{-1mm}:=\small\begin{bmatrix}
				\mathbf{N}^\top&\hspace{-3mm}\mathbf{M}^\top
			\end{bmatrix}\hspace{-2mm}\phantom{.}^\top\mathbf{G}_0^{-1}$ designates a stable NRCF of $\mathbf{G}$.
		\end{enumerate}
	\end{lemma}
	
	\begin{proof}
		For point $\mathbf{(a)}$, see the Appendix. Point $\mathbf{(b)}$ is precisely Proposition 1 in \cite{VNCF}.\vspace{-1mm}
	\end{proof}
	
	Having now the ability to express the TFMs that make up \eqref{eq:class}, we turn our attention to characterizing stabilizing controllers whose NRF implementations of type \eqref{eq:NRF_explicit} stabilize all TFMs in $\mathcal{C}_{\mathbf{G}}^\epsilon$, for a given $\epsilon\in(0,1]$. The following result is central to this section and offers the means to do just so.\vspace{-1mm}
	
	
	\begin{theorem}\label{thm:robstab}
		
		Let $\mathbf{G}\in\mathcal{R}^{p\times m}$ be given by a strongly stabilizable and detectable realization \eqref{eq:tfm} and let $F$ ensure that $A+BF-sE$ is admissible. Let also $\mathbf{G}=\mathbf{N}\mathbf{M}^{-1}$ be the stable RCF induced by F as in \eqref{eq:dublu2}, and for which a realization as in \eqref{eq:real_r} is obtained (recall Remark \ref{rem:dss2ss}), having $E_r$ invertible and $\Lambda(A_r-sE_r)\subset\mathbb{C}^-$. Let $F_r$ be the stabilizing feedback of the GCARE from \eqref{eq:GCARE} and let $\epsilon\in(0,1]$ along with $H_r\in\mathbb{R}^{m\times m}$ invertible, such that $H_r^{\top}H_r=D_r^{\top}D_r$. Then:
		
		\begin{enumerate}
			
			\item[$\mathbf{(a)}$] There exists a class of stabilizing controllers ${\mathbf{K}}\in\mathcal{R}^{m\times p}$\hspace{-1mm}, based upon a DCF over $\mathcal{RH}_\infty$ of $\mathbf{T}_{22}^\epsilon$, for the system
			\begin{align}\label{eq:gen_rob}
			\hspace{-1mm}{\mathbf{T}}^{\epsilon}\hspace{-0.5mm}:&=\hspace{-1mm}\small\left[\begin{array}{ll:l}
			0&-\epsilon\widehat{\mathbf{M}}^{-1}&\epsilon\widehat{\mathbf{M}}^{-1}\\\hdashline I&\phantom{\epsilon}-\mathbf{G}&\phantom{\epsilon}\mathbf{G}
			\end{array}\right]\hspace{-0.5mm}=\hspace{-0.5mm}\left[\begin{array}{c:c}
			{\mathbf{T}}_{11}^{\epsilon}&{\mathbf{T}}_{12}^{\epsilon}\\\hdashline{\mathbf{T}}_{21}^{\epsilon}&{\mathbf{T}}_{22}^{\epsilon}
			\end{array}\right]\nonumber\\&=\hspace{-1mm}\footnotesize\left[\begin{array}{cc|cc:c}
			A_r-sE_r&-B_rF&0&-B_r&B_r\\
			0&A-sE&0&-B&B\\\hline
			-\epsilon H_rF_r&-\epsilon H_rF&0&-\epsilon H_r&\epsilon H_r\\\hdashline
			0&C&I&-D&D
			\end{array}\right]\hspace{-1mm};
			\end{align}
			
			\item[$\mathbf{(b)}$]  Let $\mathbf{K}$ belong to the class from $\mathbf{(a)}$. If $
			\left\|\mathcal{F}_\ell({\mathbf{T}}^{\epsilon}\hspace{-0.5mm},\hspace{-0.5mm}{\mathbf{K}})\right\|_\infty\hspace{-1mm}\leq1
			$ and $\mathbf{K}$ admits an NRF implementation as in \eqref{eq:NRF_def_1}-\eqref{eq:NRF_def_2}, then the control laws from \eqref{eq:NRF_explicit} stabilize all $ \mathbf{G}_\mathbf{\Delta}\in\mathcal{C}_{\mathbf{G}}^\epsilon$.
		\end{enumerate}
	\end{theorem}
	\begin{proof}
		See the Appendix.\vspace{-1mm}
	\end{proof}

	\begin{remark}\label{rem:aug_sparse}
		The key to bypassing the feasibility of the model-matching problem tackled in Proposition \ref{prop:NRF} lies with judiciously employing Theorem \ref{thm:robstab}.
		Let our network's TFM be $\overline{\mathbf{G}}\in\mathcal{R}^{p\times m}$ and assume that the chosen NRF architecture is either infeasible or difficult to satisfy for the available DCFs over $\mathcal{RH}_\infty$ of $\overline{\mathbf{G}}$. Then, we may resort to an approximation of $\overline{\mathbf{G}}$, denoted $\mathbf{G}\in\mathcal{R}^{p\times m}$, which satisfies $\overline{\mathbf{G}}\in\mathcal{C}_{\mathbf{G}}^\epsilon$ and which is described by a DCF over $\mathcal{RH}_\infty$ that supports the desired NRF architecture. By obtaining control laws of type \eqref{eq:NRF_explicit} with the desired sparsity structure and which stabilize all $\mathbf{G}_{\mathbf{\Delta}}\in\mathcal{C}_{\mathbf{G}}^\epsilon$, these sparse control laws will also stabilize $\overline{\mathbf{G}}$. A concrete example of this design procedure will be shown in Section \ref{sec:examp}.\vspace{-1mm}
	\end{remark}

	Although we now possess the means to characterize robustly stabilizing NRF-based implementations of the controller, note that these are obtained by employing a DCF over $\mathcal{RH}_\infty$ whose realization is of the same order as that in \eqref{eq:gen_rob}. The next result shows how to obtain descriptor representations for the DCF over $\mathcal{RH}_\infty$ with the \emph{same order} as that of the network's model.
	
	\begin{proposition}\label{prop:size_reduce}
		Let the same framework, hypotheses and notation hold as in the statement of Theorem \ref{thm:robstab} and let $\mathbf{T}^\epsilon$ be defined as in \eqref{eq:gen_rob}. Then, we have that:
		\begin{enumerate}
			\item[$\mathbf{(a)}$] For any $H$ so that the pencil $A+HC-sE$ is admissible, a DCF over $\mathcal{RH}_\infty$ of $\mathbf{T}^\epsilon_{22}$ is given by
			\begin{subequations}
				\begin{align}
				&\hspace{-4mm}\footnotesize\left[\begin{array}{r:r}
				\widetilde{\mathbf Y}^\epsilon&\hspace{-1mm}-\widetilde{\mathbf X}^\epsilon\\\hdashline
				\hspace{-1mm}-\widetilde{\mathbf N}^\epsilon&\widetilde{\mathbf M}^\epsilon
				\end{array}\right]\hspace{-1mm}:=\hspace{-1.5mm}\left[\begin{array}{c|c:c}
				\hspace{-1mm}A+HC-sE\hspace{-0.5mm}&\hspace{-1mm}-B-HD&H\hspace{-1mm}\\\hline
				F&I&0\\\hdashline
				C&-D&I
				\end{array}\right]\hspace{-1mm},\hspace{-1mm}\normalsize\label{eq:dublu3}\\
				&\hspace{-2mm}\footnotesize\left[\begin{array}{r:r}
				\mathbf M^\epsilon&\mathbf X^\epsilon\\\hdashline
				\mathbf N^\epsilon&\mathbf Y^\epsilon
				\end{array}\right]\hspace{-1mm}:=\hspace{-1mm}\footnotesize\left[\begin{array}{c|c:c}
				\hspace{-1mm}A+BF-sE&B&-H\\\hline
				F&I&0\\\hdashline
				C+DF&D&I
				\end{array}\right]\normalsize;\label{eq:dublu4}
				\end{align}
			\end{subequations}
			
			\item[$\mathbf{(b)}$] For any stabilizing controller obtained using \eqref{eq:dublu3}-\eqref{eq:dublu4} and an arbitrary $\mathbf{Q}\in\mathcal{RH}_\infty^{m\times p}$, we may express
			$
			\mathcal{F}_\ell({\mathbf{T}}^{\epsilon},{\mathbf{K}})=\mathbf{T}^\epsilon_1+\mathbf{T}^\epsilon_2\mathbf{Q}\mathbf{T}^\epsilon_3,
			$
			where we have
			\begin{subequations}
				\begin{align}
				\mathbf T_{1}^\epsilon:=&\ {\mathbf{T}}_{11}^{\epsilon}+{\mathbf{T}}_{12}^{\epsilon}{\mathbf X}^\epsilon\widetilde{\mathbf M}^\epsilon{\mathbf{T}}_{21}^{\epsilon}\nonumber\\
				=&\footnotesize\left[\begin{array}{cc|cc}
				\hspace{-1mm}A_r-sE_r&\hspace{-2mm}-B_rF&0&\hspace{-2mm}-B_r\\
				0&\hspace{-2mm}A+HC-sE&H&\hspace{-2mm}-B-HD\hspace{-1mm}\\\hline
				-\epsilon H_rF_r&\hspace{-2mm}-\epsilon H_rF&0&\hspace{-2mm}-\epsilon H_r
				\end{array}\right]\hspace{-1mm},\hspace{-1mm}\label{eq:afin5}\\
				\mathbf T_{2}^\epsilon:=&\ {\mathbf{T}}_{12}^{\epsilon}\mathbf M^\epsilon=\left[\footnotesize\begin{array}{c|c}
				A_r-sE_r&B_r\\\hline -\epsilon H_rF_r&\epsilon H_r
				\end{array}\right],\label{eq:afin6}\\
				\mathbf T_{3}^\epsilon:=&\ \widetilde{\mathbf M}^\epsilon{\mathbf{T}}_{21}^{\epsilon}=\left[\footnotesize\begin{array}{c|cc}
				\hspace{-1mm}A+HC-sE&\hspace{-1mm}H&\hspace{-2mm}-B-HD\hspace{-1mm}\\\hline
				C&I&\hspace{-2mm}-D
				\end{array}\right]\hspace{-1mm}.\hspace{-1mm}\label{eq:afin7}
				\end{align}
			\end{subequations}
		\end{enumerate}
	\end{proposition}
	\begin{proof}
		See the Appendix.
	\end{proof}

	\section{Convex procedure for augmented sparsity}\label{sec:algo}


	
	\subsection{Procedure setup and norm condition reformulation}
	
	Recall that, in order to obtain sparse control laws of type \eqref{eq:NRF_explicit}, we aim to express controllers of type \eqref{eq:K_frac} for $\mathbf{Q}\in\mathcal{RH}_\infty^{m\times p}$ satisfying \eqref{eq:NRF_EMM}-\eqref{eq:fin_prop}. For robust stability, point $\mathbf{(b)}$ of Theorem \ref{thm:robstab} argues that we need only satisfy $\|\mathbf{T}^\epsilon_1+\mathbf{T}^\epsilon_2\mathbf{Q}\mathbf{T}^\epsilon_3\|_\infty\leq 1$, where $\mathbf{T}^\epsilon_1$, $\mathbf{T}^\epsilon_2$ and $\mathbf{T}^\epsilon_3$ are expressed as in \eqref{eq:afin5}-\eqref{eq:afin7}.
	
	The beginning of this section is dedicated to showing how this norm condition can be converted into \eqref{eq:prob}. Due to this being the setup of the iterative algorithm given in the sequel, this conversion will be given in an ordered sequence of steps:
	
	\emph{Step 1.} Solve \eqref{eq:NRF_EMM} for $\mathbf{Q}_0\in\mathcal{RH}_\infty^{m\times p}$ and obtain a basis $\mathbf{B}\in\mathcal{RH}_\infty^{mp\times q}$ for $\text{Ker}\small\begin{bmatrix}
		F_{\mathcal{X}}(\widetilde{\mathbf{M}}^\top\otimes I)\\F_{\widehat{\mathcal{Y}}}(\widetilde{\mathbf{N}}^\top\otimes I)
	\end{bmatrix}$ (recall Remark \ref{rem:reuse_QI_orig});
	

	\emph{Step 2.} Partition $\mathbf{B}$ via its columns, as follows
	$${\mathbf{B}}:=\left[\begin{array}{c:c:c:c:c}
	\mathbf{B}_1&\dots&\mathbf{B}_i&\dots&\mathbf{B}_q
	\end{array}\right],\ \mathbf{B}_i\in\mathcal{RH}_\infty^{mp\times 1},$$
	to obtain minimal realizations $\mathbf{B}_i=\left[\small\begin{array}{c|c}
	A^\mathbf{B}_i&B^\mathbf{B}_i\\\hline C^\mathbf{B}_i&D^\mathbf{B}_i
	\end{array}\right]$,$\forall i\in1:q$;

	\emph{Step 3.} Using these realizations, write via \eqref{eq:dublu2} a stable RCF of each $\mathbf{B}_i=\mathbf{N}_{\mathbf{B}_i}\mathbf{M}_{\mathbf{B}_i}^{-1}$, which are given explicitly by
	\begin{equation}\label{eq:B_fac}
	\begin{bmatrix}
	\mathbf{M}_{\mathbf{B}_i}\\\mathbf{N}_{\mathbf{B}_i}
	\end{bmatrix}:=\left[\small\begin{array}{c|c}
	A^\mathbf{B}_i+B^\mathbf{B}_iF^\mathbf{B}_i&B^\mathbf{B}_i\\\hline
	F^\mathbf{B}_i&1\\
	C^\mathbf{B}_i+D^\mathbf{B}_iF^\mathbf{B}_i&D^\mathbf{B}_i
	\end{array}\right]\in\mathcal{RH}_\infty^{(mp+1)\times 1},
	\end{equation}
	with $F^\mathbf{B}_i$ ensuring $\Lambda(A^\mathbf{B}_i+B^\mathbf{B}_iF^\mathbf{B}_i-sI)\subset\mathbb{C}^-$ to form
	\begin{equation}\label{eq:B_hat}
		\widehat{\mathbf{B}}:=
		\left[\begin{array}{c:c:c:c:c}
			\mathbf{N}_{\mathbf{B}_1}&\cdots&\mathbf{N}_{\mathbf{B}_i}&\cdots&\mathbf{N}_{\mathbf{B}_q}
		\end{array}\right]\in\mathcal{RH}_\infty^{mp\times q};
	\end{equation}
	
	\emph{Step 4.} Partition $\widehat{\mathbf{B}}=\left[\begin{array}{c:c:c:c:c}
	\widehat{\mathbf{B}}_1^\top&\cdots&\widehat{\mathbf{B}}_i^\top&\cdots&\widehat{\mathbf{B}}_p^\top
	\end{array}\right]^\top,$ noting that $\widehat{\mathbf{B}}_i\in\mathcal{RH}_\infty^{m\times q},$
	in order to finally define
	\begin{equation}\label{eq:B_bar}
	\overline{\mathbf{B}}:=\left[\begin{array}{c:c:c}
	\widehat{\mathbf{B}}_1&\cdots&\widehat{\mathbf{B}}_p
	\end{array}\right]\hspace{-1mm}=\hspace{-1mm}\left[\small\begin{array}{c|c}
	A_{\overline{\mathbf{B}}}&B_{\overline{\mathbf{B}}}\\\hline
	C_{\overline{\mathbf{B}}}&D_{\overline{\mathbf{B}}}
	\end{array}\right]\in\mathcal{RH}_\infty^{m\times pq}.
	\end{equation}

	\begin{remark}\label{rem:bas_alt}
		Since $\mathbf{B}_i$ are the columns of stable basis of the null space in \eqref{eq:null_NRF}, then so are $\mathbf{N}_{\mathbf{B}_i}=\mathbf{B}_i\mathbf{M}_{\mathbf{B}_i}$, having realizations of the same order as those of $\mathbf{B}_i$. Thus, $\widehat{\mathbf{B}}$ is a stable basis for the same null space and may also be used to form $\text{vec}(\widehat{\mathbf{Q}})=\widehat{\mathbf{B}}\mathbf{x},\ \forall\mathbf{x}\in\mathcal{RH}_\infty^{q\times 1}$, as in Proposition \ref{prop:NRF}.
	\end{remark}

	This concludes the setup of our procedure and we now move on to converting $\|\mathbf{T}^\epsilon_1+\mathbf{T}^\epsilon_2\mathbf{Q}\mathbf{T}^\epsilon_3\|_\infty\leq 1$ into \eqref{eq:prob}, through the explicit use of $\overline{\mathbf{B}}$. Recall that $\mathbf{Q}$ can be partitioned additively as $\mathbf{Q}=\mathbf{Q}_0+\widehat{\mathbf{Q}}\,$, with $\mathbf{Q}_0$ having been obtained in \textit{Step 1} of the setup and with $\text{vec}(\widehat{\mathbf{Q}})$ formed as in Remark \ref{rem:bas_alt}. Thus, by \eqref{eq:B_bar} in \textit{Step 4} of the setup, it is straightforward to obtain\vspace{-1mm}
	$$
	\widehat{\mathbf{Q}}=\left[\begin{array}{c:c:c}
	\widehat{\mathbf{B}}_1\mathbf{x}&\cdots&\widehat{\mathbf{B}}_p\mathbf{x}
	\end{array}\right]=\overline{\mathbf{B}}\mathcal{D}_p(\mathbf{x})\in\mathcal{RH}_\infty^{m\times p}.\vspace{-1mm}
	$$
	Defining $\mathbf{x}\in\mathcal{RH}_\infty^{q\times 1}$ as $\mathbf{x}:=\left[\small\begin{array}{c|c}
	A_x&b_x\\\hline C_x&d_x
	\end{array}\right]$,
	we may express
	\begin{equation}\label{eq:Q_hat_ss}
	\widehat{\mathbf{Q}}=\left[\small\begin{array}{cc|c}
	A_{\overline{\mathbf{B}}}&B_{\overline{\mathbf{B}}}\mathcal{D}_p(C_x)&B_{\overline{\mathbf{B}}}\mathcal{D}_p(d_x)\\
	0&\mathcal{D}_p(A_x)&\mathcal{D}_p(b_x)\\\hline
	C_{\overline{\mathbf{B}}}&D_{\overline{\mathbf{B}}}\mathcal{D}_p(C_x)&D_{\overline{\mathbf{B}}}\mathcal{D}_p(d_x)
	\end{array}\right],\normalsize
	\end{equation}
	whose realization is affine in terms of all variable matrices: $A_x$, $b_x$, $C_x$, $d_x$, and $A_{\overline{\mathbf{B}}}$ and $C_{\overline{\mathbf{B}}}$, by way of $F_i^\mathbf{B}$, for $i\in1:q$.

	It now becomes clear, in terms of \eqref{eq:prob}, that we have $\widehat{\mathbf{T}}_1=\mathbf{T}^\epsilon_1+\mathbf{T}^\epsilon_2\mathbf{Q}_0\mathbf{T}^\epsilon_3$ and $\widehat{\mathbf{T}}_{2i}=\mathbf{T}^\epsilon_2\widehat{\mathbf{Q}}_i\mathbf{T}^\epsilon_3$, where we have defined
	\begin{equation}\label{eq:Q_hat_i}
		\widehat{\mathbf{Q}}_i:=\overline{\mathbf{B}}\left[\begin{array}{c:c:c:c:c}
			\widehat e_i&\cdots&\widehat e_{jq+i}&\cdots&\widehat e_{(p-1)q+i}
		\end{array}\right],
	\end{equation}
	with $i\in 1:q,\ j\in 1:p-2$, and $\widehat e_i$ being the $i^\text{th}$ vector in the canonical basis of $\mathbb{R}^{pq\times 1}$. Moving on, the next section tackles the numerical details of satisfying the inequality from \eqref{eq:prob}.
	
	\begin{remark}\label{rem:reduce}
		The free term of the Youla Parametrization is now expressed as $\mathbf{Q}=\mathbf{Q}_0+\sum_{i=1}^{q}\mathbf{x}_i\widehat{\mathbf{Q}}_i$, for some $\mathbf{x}_i\in\mathcal{RH}_\infty^{1\times 1},\ \forall i\in1:q$. Thus, forming $\widehat{\mathbf{B}}$ from only a subset of the $q$ columns used in \eqref{eq:B_hat} may prove sufficient to solve \eqref{eq:prob}, which has the benefit of cutting down on computational costs.\vspace{-1mm}
	\end{remark}
	
	\subsection{Numerical formulation and NRF implementability}
	
	In order to formulate a numerical procedure meant to solve \eqref{eq:prob}, we first require a state-space realization of $\mathbf{T}^\epsilon_1+\mathbf{T}^\epsilon_2\mathbf{Q}\mathbf{T}^\epsilon_3$. This can be obtained by first defining the following TFM\vspace{-1mm}
	\begin{equation}\label{eq:Tf}
		\mathbf{T}^f\hspace{-1mm}:=\hspace{-1mm}\small\left[\begin{array}{c:c}
			\mathbf{T}^\epsilon_1+\mathbf{T}^\epsilon_2\mathbf{Q}_0\mathbf{T}^\epsilon_3&\mathbf{T}^\epsilon_2\\\hdashline
			\mathbf{T}^\epsilon_3&0
		\end{array}\right]\hspace{-1mm}=\hspace{-1mm}\left[\small\begin{array}{c|c:c}
			A^f&B^f_1&B^f_2\\\hline
			C^f_1&D^f_{11}&D^f_{12}\\\hdashline
			C^f_2&D^f_{21}&0
		\end{array}\right]\hspace{-1mm},\normalsize\vspace{-1mm}
	\end{equation}
	and obtaining a minimal state-space realization as in \eqref{eq:Tf}, with $\Lambda\left(A^f-sI\right)\subset\mathbb{C}^-$ due to $\mathbf{T}^f\in\mathcal{RH}_\infty$, via one of $\mathbf{Q}_0$ and via \eqref{eq:afin5}-\eqref{eq:afin7}, as per Remark \ref{rem:dss2ss}. Notice that 	
	$\mathbf{T}^\epsilon_1+\mathbf{T}^\epsilon_2\mathbf{Q}\mathbf{T}^\epsilon_3=\mathcal{F}_\ell(\mathbf{T}^f,\widehat{\mathbf{Q}})$ to get,
	via \eqref{eq:Q_hat_ss} along with the formulas in Section 10.4 of \cite{zhou}, the realization from \eqref{eq:H_CL_ss}, given on the next page. Crucially, notice that all the variable matrices which appear in the realization from \eqref{eq:H_CL_ss} do so only via \emph{affine terms}. 

	Before stating the numerical problem which will be tackled by our iterative procedure, we must ensure that the obtained controller is well-defined and can be implemented as in \eqref{eq:NRF_explicit}, via its NRF pair. As indicated in Remark \ref{rem:prop}, this is ensured by satisfying \eqref{eq:BMI1}-\eqref{eq:BMI2}, which can be written generically as\stepcounter{equation}
	\begin{equation}\label{eq:ineq_prop}
		\left(Z_1^k+Z_2^k\mathbf{Q}(\infty)Z_3^k\right)\hspace{-2mm}\phantom{.}^{\top}\hspace{-1mm}\left(Z_1^k+Z_2^k\mathbf{Q}(\infty)Z_3^k\right)\succ 0,\ \forall k\in 1:N_Z,
	\end{equation}
	where the various matrices $Z_1^k\in\mathbb{R}^{w_k\times w_k}$, $Z_2^k\in\mathbb{R}^{w_k\times m}$ and $Z_3^k\in\mathbb{R}^{p\times w_k}$ are shown explicitly in \eqref{eq:BMI1}-\eqref{eq:BMI2}. Finally, note that	$
	\mathbf{Q}(\infty)=\mathbf{Q}_0(\infty)+\sum_{i=1}^qd_{xi}\widehat{\mathbf{Q}}_i(\infty),
	$
	where we partition
	$
	\mathbf{x}(\infty)=d_x=\left[\begin{array}{c:c:c:c:c}
	d_{x1}&\cdots&d_{xi}&\cdots&d_{xq}
	\end{array}\right]^\top,\ i\in 2:q-1.
	$
	Thus, we combine \eqref{eq:prob} and \eqref{eq:ineq_prop} into our numerical problem
	\begin{equation}\label{eq:prob_x}\footnotesize
	\left\{\hspace{-1mm}\begin{array}{l}
	\left\|\mathcal{F}_\ell(\mathbf{T}^f,\widehat{\mathbf{Q}})\right\|_\infty<1,\ \widehat{\mathbf{Q}}\text{ as in }\eqref{eq:Q_hat_ss},\\[2mm]
	(\widehat{Z}_1^k)^\top\widehat{Z}_1^k+\hspace{-1mm}\sum\limits_{i=1}^{q}d_{xi}\big((\widehat{Z}_1^k)^\top\widehat{Z}_{2i}^k+(\widehat{Z}_{2i}^k)^\top\widehat{Z}_1^k\big)+\hspace{-1mm}\sum\limits_{i=1}^{q}d_{xi}^2(\widehat{Z}_{2i}^k)^\top\widehat{Z}_{2i}^k+\\+\sum\limits_{i=1}^{q-1}\sum\limits_{j=i+1}^{q}d_{xi}d_{xj}\big((\widehat{Z}_{2i}^k)^\top\widehat{Z}_{2j}^k+(\widehat{Z}_{2j}^k)^\top\widehat{Z}_{2i}^k\big)\succ 0,\ \forall k\in1:N_Z,\\
	\widehat{Z}_1^k:=Z_1^k+Z_2^k\mathbf{Q}_0(\infty)Z_3^k,\widehat{Z}_{2i}^k:=Z_2^k\widehat{\mathbf{Q}}_i(\infty)Z_3^k,\ \forall i\in1:q.
	\end{array}\right.\normalsize\hspace{-2mm}
	\end{equation}\newpage

	\begin{algorithm}\small
		\textbf{Initialization:} Solve the LMI system of \eqref{eq:prob_iter}, given on the next page, along with the equality constraint $\overline{d}_x-d_x=0$, for $\Big( A_{x}^0,b_{x}^0,C_{x}^0,d_{x}^0,\overline{d}_x^0,\left(F^\mathbf{B}_i\right)^{0},P^0,\overline{P}^0,\left(\overline{P}^D\right)^{0},P_x^0,$ $\overline{P}_x^0,\left(P^\mathbf{B}_i\right)^{0},\left(\overline{P}^\mathbf{B}_i\right)^{0}\Big)$. Using these computed variables, form $T_A^0$, $T_B^0$, $T_C^0$ as in \eqref{eq:aux4}-\eqref{eq:aux6} and then set $k=0$ along with $f^{0}=\footnotesize\left\|T_C^{0}-T_A^{0}T_B^{0}\right\|_*+\left\|I_{n_T}\right\|_*$\;
		
		\Repeat{$f^{k-1}-f^{k}<\eta_1$ or $f^k-\left\|I_{n_T}\right\|_*<\eta_2$}{
			\eIf{$k\text{ mod }2<1$}{
				Set $k=k+1$ followed by $\Theta^k=T_B-T_B^{k-1}$\;
			}{Set $k=k+1$ followed by $\Theta^k=T_A-T_A^{k-1}$\;}
			
			Solve $\mathcal{M}\left(T_A^{k-1},T_B^{k-1},\Theta^k\right)$ for $\Big(\hspace{-0.5mm} A_{x}^{k},b_{x}^{k},C_{x}^{k},d_{x}^{k},\overline{d}_x^{k},$ $(F^\mathbf{B}_i)^{k},P^{k},\overline{P}^{k},\left(\overline{P}^D\right)^{\hspace{-0.5mm}k}\hspace{-1mm},P_x^{k}, \overline{P}_x^{k},$ $(P^\mathbf{B}_i)^{k},\left(\overline{P}^\mathbf{B}_i\right)^{\hspace{-0.5mm}k}\Big)$ and use them to form $T_A^{k}$, $T_B^{k}$, $T_C^{k}$ as in \eqref{eq:aux4}-\eqref{eq:aux6}\;
			
			Compute $f^{k}:=\left\|
			T_C^{k}-T_A^{k}T_B^{k}\right\|_*+\left\|I_{n_T}\right\|_*$\;
			
		}\normalsize\vspace{2mm}\caption{Convex approach to solving \eqref{eq:prob_x}}\label{alg:iter}
	\end{algorithm}

	\phantom{}\vspace{-7mm}

	\subsection{The iterative procedure with guaranteed convergence}
	
	We now introduce the most general form (recall Remark \ref{rem:reduce}) of our convex and iterative procedure for solving \eqref{eq:prob_x}, based upon the algorithm with guaranteed convergence in 
	\cite{BMI2LMI}.

	\begin{figure*}[b]
		\hrulefill
		\begin{equation}\label{eq:H_CL_ss}
			\mathbf{T}^\epsilon_1+\mathbf{T}^\epsilon_2\mathbf{Q}\mathbf{T}^\epsilon_3=\left[\begin{array}{c|c}
				\overline{A}&\overline{B}\\\hline
				\overline{C}&\overline{D}
			\end{array}\right]=\left[\footnotesize\begin{array}{ccc|c}
				A^f+B_2^fD_{\overline{\mathbf{B}}}\mathcal{D}_p(d_x)C_2^f&B_2^fC_{\overline{\mathbf{B}}}&B_2^fD_{\overline{\mathbf{B}}}\mathcal{D}_p(C_x)&B_1^f+B_2^fD_{\overline{\mathbf{B}}}\mathcal{D}_p(d_x)D_{21}^f\\
				B_{\overline{\mathbf{B}}}\mathcal{D}_p(d_x)C_2^f&A_{\overline{\mathbf{B}}}&B_{\overline{\mathbf{B}}}\mathcal{D}_p(C_x)&B_{\overline{\mathbf{B}}}\mathcal{D}_p(d_x)D_{21}^f\\
				\mathcal{D}_p(b_x)C_2^f&0&\mathcal{D}_p(A_x)&\mathcal{D}_p(b_x)D_{21}^f\\\hline
				C_1^f+D_{12}^fD_{\overline{\mathbf{B}}}\mathcal{D}_p(d_x)C_2^f&D_{12}^fC_{\overline{\mathbf{B}}}&D_{12}^fD_{\overline{\mathbf{B}}}\mathcal{D}_p(C_x)&D_{11}^f+D_{12}^fD_{\overline{\mathbf{B}}}\mathcal{D}_p(d_x)D_{21}^f
			\end{array}\right]\normalsize\tag{26}
		\end{equation}
		\hrulefill
		\begin{equation}\label{eq:prob_iter}
		\mathcal{M}\left(X,Y,\Theta\right):=\footnotesize\left\{\begin{array}{l}
		\min\limits_{A_{x},b_{x},C_{x},d_{x},\overline{d}_x,F^\mathbf{B}_i,P,\overline{P},\overline{P}^D,P_x,\overline{P}_x,P^\mathbf{B}_i,\overline{P}^\mathbf{B}_i}\left\|\begin{bmatrix}
		T_C+XY-T_AY-XT_B&T_A-X\\T_B-Y&I_{n_T}
		\end{bmatrix}\right\|_*,\\
		\text{s.t.}\left\{
		\begin{array}{l}
		(\widehat{Z}_1^k)^\top\widehat{Z}_1^k+\sum_{i=1}^{q}d_{xi}\big((\widehat{Z}_1^k)^\top\widehat{Z}_{2i}^k+(\widehat{Z}_{2i}^k)^\top\widehat{Z}_1^k\big)+\sum_{i=1}^{q}\overline{P}^D_{ii}(\widehat{Z}_{2i}^k)^\top\widehat{Z}_{2i}^k+\\
		\qquad\qquad\qquad\qquad\qquad\qquad\qquad\qquad\quad+\sum_{i=1}^{q-1}\sum_{j=i+1}^{q}\overline{P}^D_{ij}\big((\widehat{Z}_{2i}^k)^\top\widehat{Z}_{2j}^k+(\widehat{Z}_{2j}^k)^\top\widehat{Z}_{2i}^k\big)\succ 0,\ \forall k\in1:N_Z,\\
		\overline{P}^D = \big(\overline{P}^D\big)^\top,\ P=P^{\top}\succ 0,\ -G\succ 0,\ P_x=P_x^\top\succ 0,\ -2\text{sym}(\overline{P}_x)\succ 0,\\ P_i^\mathbf{B}=(P_i^\mathbf{B})^{\top}\succ 0,\ \forall i\in 1:q,\ -2\text{sym}\big(P_i^\mathbf{B}\left(A_i^\mathbf{B}\right)^\top+\overline{P}_i^\mathbf{B}\left(B_i^\mathbf{B}\right)^\top\big)\succ 0,\ \forall i\in 1:q,\ \Theta=0.\\
		\end{array}
		\right.\end{array}\right.\hspace{-3mm}
		\normalsize\tag{30}
		\end{equation}\vspace{-3mm}
		\hrulefill
	\end{figure*}

	\begin{theorem}\label{thm:algo}
		
		Given the realization from \eqref{eq:H_CL_ss} along with two tolerance values $0<\eta_1,0<\eta_2\ll1$, define the following:
		\begin{subequations}
			\begin{equation}\label{eq:aux1}
			\small\begin{array}{ll}
			\overline{A}^f:=\begin{bmatrix}
			A^f&0&0\\0&0&0\\0&0&0
			\end{bmatrix},&\overline{C}_2^f:=\begin{bmatrix}
			C_2^f&0&0\\
			0&I&0\\
			0&0&I
			\end{bmatrix},\\
			\overline{B}^f_{1\phantom{2}}\hspace{-1mm}:=\begin{bmatrix}\left(B_{1\phantom{2}}^f\right)\hspace{-2mm}\phantom{.}^\top&\hspace{-2mm}0&\hspace{-2mm}0\end{bmatrix}^\top\hspace{-1mm},&	\overline{D}_{21}^f\hspace{-1mm}:=\begin{bmatrix}
			\left(D_{21}^f\right)\hspace{-2mm}\phantom{.}^\top&\hspace{-2mm}0&\hspace{-2mm}0
			\end{bmatrix}^\top\hspace{-1mm},
			\end{array}\normalsize
			\end{equation}\vspace{-1mm}
			\begin{equation}\label{eq:aux2}
			G:=\footnotesize\begin{bmatrix}
			2\text{sym}\left(P\overline{A}^f+\overline{P}\overline{C}_2^f\right) &
			P\overline{B}_1^f+\overline{P}\overline{D}_{21}^f & \overline{C}^{\top}\\
			\left(\overline{B}_1^f\right)^{\top}P+\left(\overline{P}\overline{D}_{21}^f\right)^\top&- I& {\overline{D}}^{\top}\\
			\overline{C}&\overline{D}&- I
			\end{bmatrix},\normalsize
			\end{equation}\vspace{-1mm}
			\begin{equation}\label{eq:aux3}
			A_S:=\small\begin{bmatrix}
			B_2^fD_{\overline{\mathbf{B}}}\mathcal{D}_p(d_x)&B_2^fC_{\overline{\mathbf{B}}}&B_2^fD_{\overline{\mathbf{B}}}\mathcal{D}_p(C_x)\\
			B_{\overline{\mathbf{B}}}\mathcal{D}_p(d_x)&A_{\overline{\mathbf{B}}}&B_{\overline{\mathbf{B}}}\mathcal{D}_p(C_x)\\
			\mathcal{D}_p(b_x)&0&\mathcal{D}_p(A_x)
			\end{bmatrix},\normalsize\vspace{-1mm}
			\end{equation}
			\begin{equation}
			\label{eq:aux4}
			\hspace{0mm}T_A:=\mathcal{D}\left(P_1^\mathbf{B},\dots,P_q^\mathbf{B},P_x,\overline{d}_x,P,0\right)\in\mathbb{R}^{p_T\times n_T},\normalsize\vspace{-1mm}
			\end{equation}
			\begin{equation}
			\label{eq:aux5}
			T_B:=\mathcal{D}\left(\left(F_1^\mathbf{B}\right)^\top,\dots,\left(F_q^\mathbf{B}\right)^\top,A_x^\top,d_x^\top,A_S,0\right)\normalsize\in\mathbb{R}^{n_T\times m_T},\vspace{-1mm}
			\end{equation}
			\begin{equation}
			\label{eq:aux6}
			\hspace{-4mm}T_C:=\mathcal{D}\left(\overline{P}_1^\mathbf{B},\dots,\overline{P}_q^\mathbf{B},\overline{P}_x,\overline{P}^D,\overline{P},\overline{d}_x-d_x\right)\normalsize\in\mathbb{R}^{p_T\times m_T}.\vspace{-1mm}
			\end{equation}
		\end{subequations}
		Then, we have that:
		\begin{enumerate}
			\item[$\mathbf{(a)}$] If the problem from \eqref{eq:prob_x} is feasible, then a solution can be found by the iterative procedure with guaranteed convergence from Algorithm \ref{alg:iter} which involves the convex optimization problem 
			from \eqref{eq:prob_iter}, on the next page;
			
			\item[$\mathbf{(b)}$] If, at the proposed iteration's termination, we have that $\left\|T_C^{k}-T_A^{k}T_B^{k}\right\|_*<\eta_2$, then $A_{x}^{k}$, $b_{x}^{k}$, $C_{x}^{k}$, $d_{x}^{k}$ and $(F_i^\mathbf{B})^{k}$ can be used to form $\widehat{\mathbf{Q}}$ as in \eqref{eq:B_fac}-\eqref{eq:Q_hat_ss}.
		\end{enumerate}
	\end{theorem}\vspace{-3mm}

	\begin{figure*}[b]
		\centering
		\begin{tikzpicture}[scale=.75,every node/.style={transform shape}]
			\node(n1)at(0,0){};
			\node(nn1)[circle,draw,minimum height=8,right of=n1, node distance=8em]{};
			\node(n3)[rounded corners=3, minimum height=0.5cm, minimum width=0.5cm,draw, thick,right of=nn1, node distance=6em]{$ \begin{bmatrix}
					\mathbf{\Phi}_{\mathbf{K}}&\mathbf{\Gamma}_{\mathbf{K}}
				\end{bmatrix} $};
			\node(n4)[right of=n3, node distance=6em]{$\bullet$};
			\node(n5)[circle,draw,minimum height=8,right of=n4, node distance=6em]{};
			\node(n6)[right of=n5, node distance=4em]{};
			\node(n7)[rounded corners=3, minimum height=0.5cm, minimum width=0.5cm,draw, thick,right of=n6, node distance=4em]{$ \begin{bmatrix}
					\mathbf{\Phi}_{\mathbf{K}}&\mathbf{\Gamma}_{\mathbf{K}}
				\end{bmatrix} $};
			\node(n8)[right of=n7, node distance=6em]{$\bullet$};
			\node(n9)[circle,draw,minimum height=8,right of=n8, node distance=6em]{};
			\node(n10)[right of=n9, node distance=5em]{};
			
			\draw[-latex,line width=1.5pt](n1)--(nn1) node[midway, above]{\hspace{-1em}$\mathbf{u}_{((i-1)\text{ mod }\ell)+1}$};
			\draw[-latex,line width=1.5pt](n1)--(nn1) node[very near end, below]{$+$};
			\draw[-latex,line width=1.5pt](nn1)--(n3) ;
			\draw[-latex, line width=1.5pt](n3)--(n5) node[midway, above]{$\mathbf{u}_{(i\text{ mod }\ell)+1}$};
			\draw[-latex,line width=0pt](n4)--(n5) node[near end, below]{$+$};
			\draw[-latex,line width=1.5pt](n5)--(n7);
			\draw[-latex,line width=1.5pt](n7)--(n9) node[midway, above]{$\mathbf{u}_{((i+1)\text{ mod }\ell)+1}$};
			\draw[-latex,line width=0pt](n8)--(n9) node[near end, below]{$+$};
			\draw[-latex,line width=1.5pt](n9)--(n10);
			
			\node(nn21)[below of=nn1, node distance=3em]{$\mathbf{b}_{((i-1)\text{ mod }\ell)+1}$};
			\node(n25)[below of=n5, node distance=3em]{$\mathbf{b}_{(i\text{ mod }\ell)+1}$};
			\node(n29)[below of=n9, node distance=3em]{$\mathbf{b}_{((i+1)\text{ mod }\ell)+1}$};
			
			\draw[-latex,line width=1.5pt](nn21)--(nn1) node[ near end, right]{$+$};
			\draw[-latex,line width=1.5pt](n25)--(n5) node[ near end, right]{$+$};
			\draw[-latex,line width=1.5pt](n29)--(n9) node[ near end, right]{$+$};
			
			\node(n11)[below of=n1, node distance=5.5em]{};
			\node(nn11)[rounded corners=3, minimum height=0.5cm, minimum width=0.75cm,draw, thick,below of=nn1, node distance=5.5em]{$\overline{\mathbf{G}}_y$};
			\node(n24)[rounded corners=3, minimum height=0.5cm, minimum width=0.75cm,draw, thick,below of=n4, node distance=2.5em]{$\overline{\mathbf{G}}_u$};
			\node(n14)[circle,draw,minimum height=8,below of=n4, node distance=5.5em]{};
			\node(n15)[below of=n5, node distance=5.5em]{$\bullet$};
			\node(n28)[rounded corners=3, minimum height=0.5cm, minimum width=0.75cm,draw, thick,below of=n8, node distance=2.5em]{$\overline{\mathbf{G}}_u$};
			\node(n18)[circle,draw,minimum height=8, below of=n8, node distance=5.5em]{};
			\node(n19)[below of=n9, node distance=5.5em]{$\bullet$};
			\node(n20)[below of=n10, node distance=5.5em]{};
			
			\node(n26)[rounded corners=3, minimum height=0.5cm, minimum width=0.75cm,draw, thick,below of=n6, node distance=5.5em]{$\overline{\mathbf{G}}_y$};

			\draw[-latex,line width=1.5pt](n11)--(nn11) node[midway, above]{$\mathbf{y}_{((i-1)\text{ mod }\ell)+1}$};
			\draw[-latex,line width=1.5pt](nn11)--(n14)node[very near end, above]{$+$};
			\draw[-latex,line width=1.5pt](n14)--(n26) node[midway, above]{\hspace{10mm}$\mathbf{y}_{(i\text{ mod }\ell)+1}$};
			\draw[-latex,line width=1.5pt](n26)--(n18) node[very near end, above]{$+$};
			\draw[-latex,line width=1.5pt](n18)--(n20) node[midway, above]{\hspace{-0.75em}$\mathbf{y}_{((i+1)\text{ mod }\ell)+1}$};
			
			\draw[-latex,line width=1.5pt](20em,0)--(n24)--(n14)node[ near end, right]{$+$};
			\draw[-latex,line width=1.5pt](40em,0)--(n28)--(n18)node[ near end, right]{$+$};
			
			\draw[line width=1.5pt](26em,-5.5em)--(26em,-7.5em)--(14em,-7.5em)--(14em,-6em);
			\draw[line width=1.5pt] (13.95em,-6em) arc[start angle=-90, end angle=90,radius=0.5em];
			\draw[-latex, line width=1.5pt](14em,-5em)--(n3);
			
			\draw[line width=1.5pt](46em,-5.5em)--(46em,-7.5em)--(34em,-7.5em)--(34em,-6em);
			\draw[line width=1.5pt] (33.95em,-6em) arc[start angle=-90, end angle=90,radius=0.5em];
			\draw[-latex, line width=1.5pt](34em,-5em)--(n7);
			
		\end{tikzpicture}
		\caption{Interconnection between the network's subsystems and the distributed subcontrollers}\label{fig:connect}
	\end{figure*}

	\begin{proof}
		For point $\mathbf{(a)}$, see the Appendix. Point $\mathbf{(b)}$ follows directly from the fact that $\left\|T_C^{k}-T_A^{k}T_B^{k}\right\|_*<\eta_2\ll1$ indicates that the bilinear equality constraint belonging to the problem (given in the Appendix) that is equivalent to \eqref{eq:prob_x} has been satisfied for a feasible tuple, which designates a solution.
	\end{proof}\vspace{-2mm}

	\section{Numerical example}\label{sec:examp}
	
	\subsection{Design procedure}
	
	Consider a set of $\ell=20$ subsystems which are interconnected in a network with a ring topology, as depicted in Fig.~\ref{fig:connect}.
	The input-output model of each subsystem can be written as\vspace{-2mm}
	\begin{equation*}
	\mathbf{y}_{(i\text{ mod }\ell)+1}=\overline{\mathbf{G}}_y\mathbf{y}_{((i-1)\text{ mod }\ell)+1}+\overline{\mathbf{G}}_u\mathbf{u}_{(i\text{ mod }\ell)+1}, \forall i\in1:\ell,\vspace{-1mm}
	\end{equation*}
	with $\overline{\mathbf{G}}_y(s):=\tiny\left[\begin{array}{c|c}
	A_{y}-sE_{y}&B_{y}\\\hline C_{y}&0
	\end{array}\right]=\tiny\left[\begin{array}{rr|r}
	-1-s&0&1\\0&-1&1\\\hline1&1&0
	\end{array}\right]$ and $\overline{\mathbf{G}}_u(s):=\tiny\left[\begin{array}{c|c}
	A_{u}-sE_{u}&B_{u}\\\hline C_{u}&0
	\end{array}\right]=\tiny\left[\begin{array}{rr|r}
	1-s&0&11\\0&-1&4\\\hline1&1&0
	\end{array}\right]$. 
	Define now $\Xi:\mathbb{R}\rightarrow\mathbb{R}^{\ell\times \ell}$, $\Xi(\kappa):=\mathcal{D}_\ell(\kappa)\footnotesize\begin{bmatrix}
	O_{1,\ell-1}& 1\\I_{\ell-1}&O_{\ell-1,1}
	\end{bmatrix}$, to get that\stepcounter{equation}
	\begin{equation}\label{eq:net_orig}\normalsize
	\overline{\mathbf{G}}(s)=\tiny\left[\begin{array}{cc|c}
	\hspace{-2mm}\mathcal{D}_\ell(A_{y}\hspace{-1mm}-\hspace{-1mm}sE_{y})\hspace{-1mm}+\hspace{-1mm}\mathcal{D}_\ell(B_{y})\Xi(1)\mathcal{D}_\ell(C_{y})&\hspace{-2mm}\mathcal{D}_\ell(B_{y})\Xi(1)\mathcal{D}_\ell(C_{u})\hspace{-1mm}&O_{2\ell,\ell}\\O_{2\ell}&\mathcal{D}_\ell(A_{u}-sE_{u})&\mathcal{D}_\ell(B_{u})\\\hline \mathcal{D}_\ell(C_{y})&\mathcal{D}_\ell(C_{u})&O_\ell
	\end{array}\right]\normalsize
	\end{equation} 
	is the network's TFM, which is improper, having a strongly stabilizable and detectable realization and whose resulting descriptor vector is the concatenation of the descriptor vectors belonging to the realizations of all $\overline{\mathbf{G}}_y$ and $\overline{\mathbf{G}}_u$ subsystems.
	
	We aim to obtain a control law, for $\mathbf{\Phi}_{\mathbf{K}},\mathbf{\Gamma}_{\mathbf{K}}\in\mathcal{R}_p^{1\times 1}$, with\vspace{-2mm}
	\begin{equation}\label{eq:struc_des}
	\mathbf{u}_{(i\text{ mod }\ell)+1}=\mathbf{\Phi}_{\mathbf{K}}\mathbf{u}_{((i-1)\text{ mod }\ell)+1}+\mathbf{\Gamma}_{\mathbf{K}}\mathbf{y}_{(i\text{ mod }\ell)+1},\ \forall i\hspace{-0.25mm}\in\hspace{-0.25mm}1\hspace{-0.25mm}:\hspace{-0.25mm}\ell.\vspace{-1mm}
	\end{equation}
	Then, approximate $\overline{\mathbf{G}}(s)$ with $\mathbf{G}(s):=\mathcal{D}_\ell(\mathbf{\Psi})\Omega$, where $\Omega:=\Xi(2)+I_\ell$ and 
	$
	\mathbf{\Psi}(s):=\left[\footnotesize\begin{array}{c|c}
	A_\mathbf{\Psi}-sE_\mathbf{\Psi}&B_\mathbf{\Psi}\\\hline
	C_\mathbf{\Psi}&D_\mathbf{\Psi}
	\end{array}\right]=\tiny\left[\begin{array}{rr|r}
	1&-s&0\\0&1&-1\\\hline1&0&4
	\end{array}\right]\hspace{-1mm}.
	$
	Note that the latter realization is strongly stabilizable and detectable, such that $F_\mathbf{\Psi}=\begin{bmatrix}
	1&5
	\end{bmatrix}$ and $H_\mathbf{\Psi}^\top=\begin{bmatrix}
	-5&-1
	\end{bmatrix}$ are admissible feedbacks for it. Thus, we are able to express
	\begin{equation}\label{eq:net_alt}
	\mathbf{G}(s)=\left[\footnotesize\begin{array}{c|c}
	\mathcal{D}_\ell(A_\mathbf{\Psi}-sE_\mathbf{\Psi})&\mathcal{D}_\ell(B_\mathbf{\Psi})\Omega\\\hline\mathcal{D}_\ell(C_\mathbf{\Psi})&\mathcal{D}_\ell(D_\mathbf{\Psi})\Omega
	\end{array}\right],\normalsize
	\end{equation}
	with $F=\Omega^{-1}\mathcal{D}_\ell(F_\mathbf{\Psi})$ and $H=\mathcal{D}_\ell(H_\mathbf{\Psi})$ being admissible feedbacks. Obtaining stable NRCFs for \eqref{eq:net_orig} and \eqref{eq:net_alt} via Lemma \ref{lem:NRCF}, we use them to get the maximum stability radius $b_{opt}> 0.9925$ of $\mathbf{G}$ (see \cite{gap}) and to compute an upper bound for some $\overline{\mathbf{Q}}\in\mathcal{RH}_\infty^{\ell\times\ell}$ (see Chapter 8 of \cite{BF}) denoted $\mu(\overline{\mathbf{Q}})< 0.5609$ of the directed gap metric between $\mathbf{G}$ and $\overline{\mathbf{G}}$, as given in $(4)$ from \cite{gap}. Then, we set $\epsilon=0.7>\mu(\overline{\mathbf{Q}})$ and we get, by the same arguments as in the proof of Lemma 2 from \cite{gap} applied for $\overline{\mathbf{Q}}\in\mathcal{RH}_\infty^{\ell\times\ell}$, that $\overline{\mathbf{G}}\in \mathcal{C}_{\mathbf{G}}^{\epsilon}$. Note that $\mu(\overline{\mathbf{Q}})<1$ implies $\det\overline{\mathbf{Q}}\not\equiv0$. Otherwise, $\exists\ \mathbf{v}\in\text{Ker}\,\overline{\mathbf{Q}}\cap\mathcal{RH}_\infty^{\ell\times1}$ with $\mathcal{H}_2$ norm equal to $1$ which can be used to obtain that $\mu(\overline{\mathbf{Q}})\geq1$.
	
	Use now the realization from \eqref{eq:net_alt} and $F$ to compute a stable RCF as in \eqref{eq:real_r}. With this stable RCF and $H$, employ Proposition \ref{prop:size_reduce} to compute $(\mathbf{N}^{\epsilon},\widetilde{\mathbf{N}}^{\epsilon},\mathbf{M}^{\epsilon},\widetilde{\mathbf{M}}^{\epsilon},\mathbf{X}^{\epsilon},\widetilde{\mathbf{X}}^{\epsilon},\mathbf{Y}^{\epsilon},\widetilde{\mathbf{Y}}^{\epsilon})$ via \eqref{eq:dublu3}-\eqref{eq:dublu4} and $\mathbf{T}_1^{\epsilon}$, $\mathbf{T}_2^{\epsilon}$ and $\mathbf{T}_3^{\epsilon}$ as in \eqref{eq:afin5}-\eqref{eq:afin7}. Then,
	\begin{subequations}\vspace{-2mm}
		\begin{equation}\small
		\begin{array}{c}\mathbf K=(\Omega\widetilde{\mathbf {Y}}^{\epsilon}+\widetilde{\mathbf{Q}}\widetilde{\mathbf {N}}^{\epsilon})^{-1}(\Omega\widetilde{\mathbf {X}}^{\epsilon}+\widetilde{\mathbf{Q}}\widetilde{\mathbf {M}}^{\epsilon}),\end{array}\normalsize\label{eq:K_sparse}\vspace{-2mm}
		\end{equation}
		\begin{equation}\small
		\begin{array}{c}\mathcal{F}_\ell({\mathbf{T}}^{\epsilon},{\mathbf{K}})=\mathbf{T}^{\epsilon}_1+(\mathbf{T}^{\epsilon}_2\Omega^{-1})\widetilde{\mathbf{Q}}\mathbf{T}^{\epsilon}_3,\end{array}\normalsize\label{eq:CL_sparse}\vspace{-2mm}
		\end{equation}
	\end{subequations}
	having defined $\widetilde{\mathbf{Q}}:=\Omega\mathbf{Q}$. Note that $(\mathbf{N}^{\epsilon}\Omega^{-1},\widetilde{\mathbf{N}}^{\epsilon},\mathbf{M}^{\epsilon}\Omega^{-1},$ $\widetilde{\mathbf{M}}^{\epsilon},\mathbf{X}^{\epsilon},\Omega\widetilde{\mathbf{X}}^{\epsilon},\mathbf{Y}^{\epsilon},\Omega\widetilde{\mathbf{Y}}^{\epsilon})$ is also a DCF over $\mathcal{RH}_\infty$ of $\mathbf{T}^{\epsilon}_{22}=\mathbf{G}$, as all 8 TFMs are stable and they satisfy \eqref{eq:bez}, with the added benefit of $\Omega\widetilde{\mathbf{Y}}^{\epsilon},\widetilde{\mathbf{N}}^{\epsilon}\in\widehat{\mathcal{S}}_{(\Xi(1)+I_\ell)}$ and of $\Omega\widetilde{\mathbf{X}}^{\epsilon},\widetilde{\mathbf{M}}^{\epsilon}\in\widehat{\mathcal{S}}_{I_\ell}$.
	
	\begin{table*}[b]\vspace{-4mm}
		\centering
		\begin{tabular}{|c|c|c|c|c|}
			\hline
			Employed procedure&Guaranteed convergence&Runtime&Solution&$\left\|
			T_C-T_AT_B\right\|_*$  at convergence\\\hline
			Alg. 1 (Alg. 1 in \cite{BMI2LMI})&Yes&$18.41$ sec&$\widetilde{\mathbf{Q}}(s)=\mathcal{D}_\ell(5.9844)$& $\ 2.8\times10^{-12}$\\\hline Alg. 2 in \cite{BMI2LMI}&No&$20.49$ sec&$\widetilde{\mathbf{Q}}(s)=\mathcal{D}_\ell(5.9844)$&$2.1\times10^{-9}$\\\hline Alg. 1 in \cite{ADMM}&Yes&timed out after $900$ sec&$\widetilde{\mathbf{Q}}(s)=\mathcal{D}_\ell(6.0143)$ at timeout&$2.3\times10^{2}$ at timeout\\\hline
			Alg. 1 in \cite{IRM}&Yes&timed out after $900$ sec&$\widetilde{\mathbf{Q}}(s)=\mathcal{D}_\ell(5.7355)$ at timeout&$3.2\times10^{1}$ at timeout\\\hline
		\end{tabular}
		\caption{Comparison between algorithms which solve convex relaxations of \eqref{eq:prob_x}}
		\label{tab:num}
	\end{table*}\normalsize
	
	We will employ this new DCF over $\mathcal{RH}_\infty$ to form the controller as in \eqref{eq:K_sparse} and optimize the $\mathcal{H}_\infty$ norm of \eqref{eq:CL_sparse}. 
	The control laws in \eqref{eq:struc_des} can be obtained from a controller's NRF pair with $\mathbf{\Phi}=\mathcal{D}_\ell(\mathbf{\Phi}_\mathbf{K})\Xi(1)\in\widehat{\mathcal{S}}_{\Xi(1)}$ and $\mathbf{\Gamma}=\mathcal{D}_\ell(\mathbf{\Gamma}_\mathbf{K})\in\widehat{\mathcal{S}}_{I_\ell}$. By Proposition \ref{prop:NRF}, a solution to \eqref{eq:NRF_EMM} is $\widetilde{\mathbf{Q}}_0=0$ and note that a stable basis for the null-space from \eqref{eq:null_NRF} is expressed as in \eqref{eq:B_hat} with $q=\ell$ and $\mathbf{B}_i=\mathbf{N}_{\mathbf{B}_i}=\widetilde e_{1+(\ell+1)(i-1)},\ \forall i\in1:\ell$, where $\widetilde e_i$ is the $i^\text{th}$ vector of the canonical basis of $\mathbb{R}^{\ell^2\times 1}$.
	
	We now run Algorithm \ref{alg:iter} with MOSEK \cite{mosek}, called through MATLAB via YALMIP \cite{YALMIP}. 
	A comparison with other techniques from literature is given in Tab. \ref{tab:num}, located at the bottom of the next page, and their computational performance will be discussed in the next subsection. Taking $\widetilde{\mathbf{Q}}(s)=\mathcal{D}_\ell(5.9844)$ produces, $\forall i\in1:\ell$, the distributed control laws of type \eqref{eq:struc_des}\small
	\begin{equation*}
	\mathbf{u}_{(i\text{ mod }\ell)+1}=-2\mathbf{u}_{((i-1)\text{ mod }\ell)+1}+\frac{64.11s+257.4}{s+4}\mathbf{y}_{(i\text{ mod }\ell)+1}.
	\end{equation*}\normalsize
	\begin{remark}
		Let $\widetilde{\mathbf{K}}:=(I-\widetilde{\mathbf{\Phi}})^{-1}\widetilde{\mathbf{\Gamma}}$, where $\widetilde{\mathbf{\Phi}}:=-\Xi(2)$ and $\widetilde{\mathbf{\Gamma}}:=64I_\ell$, and notice that $\small\begin{bmatrix}
		\widetilde{\mathbf{\Phi}}&\widetilde{\mathbf{\Gamma}}
		\end{bmatrix}$ internally stabilizes $\begin{bmatrix}I&\overline{\mathbf{G}}\phantom{.}^\top\end{bmatrix}\hspace{-2mm}\phantom{.}^\top$. Then, the distributed implementation \eqref{eq:NRF_explicit} of the approximated distributed controller $\widetilde{\mathbf{K}}$ internally stabilizes $\overline{\mathbf{G}}$ even in the presence of communication disturbance (see \cite{NRF} and recall $\mathbf{b}_{(i\text{ mod }\ell)+1}$ from Fig. \ref{fig:connect}).
		Moreover, the control laws from \eqref{eq:NRF_explicit} implemented with either $(\mathbf{\Phi},\mathbf{\Gamma})$ or $(\widetilde{\mathbf{\Phi}},\widetilde{\mathbf{\Gamma}})$ stabilize all $\overline{\mathbf{G}}_\mathbf{\Delta}\in\mathcal{C}_{\overline{\mathbf{G}}}^\delta$ and $\delta=0.8968$, indicating satisfactory robustness.

	\end{remark}
	
	\subsection{Computational performance}\vspace{-1mm}
	
	We conclude this section by presenting a comparative discussion of the results showcased in Tab. \ref{tab:num}. With respect to our proposed procedure, inspired by \cite{BMI2LMI}, we may state that:
	
	\begin{enumerate}
		\item[1.] Algorithm 2 from \cite{BMI2LMI} is slightly more computationally demanding, due to optimizing over all decision variables during each iteration. However, this extra degree of freedom comes at the major cost of guaranteed convergence.
		
		\item[2.] Although the individual iterations of Algorithm 1 from \cite{ADMM} are significantly less costly and convergence is initially quite rapid, the latter tapers off on later iterations, similarly to Fig. 2 in \cite{ADMM}. Convergence can be sped up by the judicious choice of $\rho_1$ and $\rho_2$ form $(3.4)$ of \cite{ADMM}, yet our approach bypasses this empiric decision via the benefits of optimizing the trace heuristic (see \cite{rank}).
		
		\item[3.] Algorithm 1 from \cite{IRM} is based upon the same trace optimization heuristic proposed in \cite{rank} as our procedure, yet it requires an \emph{explicit} eigenvalue decomposition and orthonormal eigenvector computation at every iteration. For large-scale problems (such as our numerical example) this may prove unreliable, with the accumulation of computational errors noticeably hampering convergence.
	\end{enumerate} \vspace{-5mm}
	
	\section{Conclusion}\label{sec:outro}
	
	In this paper, we have shown that the distributed control of a network (having a possibly improper TFM) can be tackled by imposing constraints upon affine expressions of the Youla parameter. A procedure is given on how to relax this problem, which reduces to solving a structurally-constrained $\mathcal{H}_\infty$ norm contraction. The latter is approached through a convex and iterative optimization algorithm with guaranteed convergence.

	\vspace{-5mm}



	\section*{Appendix}

	\textbf{Proof of Proposition \ref{prop:NRF}}: Let there exist $\mathbf{Q}_0\in\mathcal{RH}_\infty^{m\times p}$ so that $\widetilde{\mathbf {X}}+\mathbf Q_0\widetilde{\mathbf {M}}\in{\widehat{\mathcal{S}}}_\mathcal{X}$ and $\widetilde{\mathbf {Y}}+\mathbf Q_0\widetilde{\mathbf {N}}\in{\widehat{\mathcal{S}}}_{\widehat{\mathcal{Y}}}$. They are equivalent to $F_{{\mathcal{X}}}\text{vec}(\widetilde{\mathbf {X}}+I\mathbf Q_0\widetilde{\mathbf {M}})\equiv0$ and $F_{\widehat{\mathcal{Y}}}\text{vec}(\widetilde{\mathbf {Y}}+I\mathbf Q_0\widetilde{\mathbf {N}})\equiv0$. Using the properties of the vectorization operator (see Lemma 1 in \cite{QI_orig}), we retrieve \eqref{eq:NRF_EMM}. Pick any $\widehat{\mathbf{Q}}\in\mathcal{RH}_\infty^{m\times p}$ which satisfies \eqref{eq:null_NRF} and note that, when replacing $\mathbf{Q}_0$ with $\mathbf{Q}:=\mathbf{Q}_0+\widehat{\mathbf{Q}}$ in \eqref{eq:NRF_EMM}, the identity with 0 from \eqref{eq:NRF_EMM} will hold. Also, ensuring \eqref{eq:mid_prop} is sufficient for the controller from \eqref{eq:K_frac} to be well-posed, while ensuring \eqref{eq:fin_prop} is sufficient for $\mathbf{\Gamma}$ and $\mathbf{\Phi}$ to be both well-posed and proper. Finally, the sparsity structures of $\mathbf{\Gamma}$ and $\mathbf{\Phi}$ follow from those of $\widetilde{\mathbf {X}}+\mathbf Q\widetilde{\mathbf {M}}$ and $\widetilde{\mathbf {Y}}+\mathbf Q\widetilde{\mathbf {N}}$, respectively, by the way they are defined in \eqref{eq:NRF_def_1}-\eqref{eq:NRF_def_2}. \qed
	
	\textbf{Proof of Lemma \ref{lem:NRCF}:} To prove point $\mathbf{(a)}$, define $\widehat{A}_r:=E_r^{-1}A_r$, $\widehat{B}_r:=E_r^{-1}B_r$, $\widehat{X}_r:=E_r^{\top}X_rE_r$ to rewrite \eqref{eq:GCARE} as\vspace{-2mm}
	\begin{equation}\label{eq:CTARE}
	\footnotesize\begin{array}{l}
	\widehat X_r\widehat A_r+\widehat A_r^{\top}\widehat X_r+C_r^{\top}C_r-(\widehat X_r\widehat B_r+C_r^{\top}D_r)\times\\
	\hspace{20mm}\times(D_r^{\top}D_r)^{-1}(\widehat B_r^{\top}\widehat X_r+D_r^{\top}C_r)=0,
	\end{array}\normalsize\vspace{-2mm}
	\end{equation}
	which is a standard continuous-time algebraic Riccati equation (see Chapter 13 in \cite{zhou}). Recall now that $\Lambda(A_r-sE_r)\subset\mathbb{C}^-$ and, thus, $\Lambda(\widehat A_r-sI)\subset\mathbb{C}^-$. Then, both $\small\begin{bmatrix}
	\widehat A_r-sI&\widehat B_r
	\end{bmatrix}$ and $\small\begin{bmatrix}
	\widehat A_r^{\top}-sI& C_r^{\top}
	\end{bmatrix}$ have full row rank $\forall s\in{\mathbb{C}}\backslash\mathbb{C}^-$ which, by section 3.2 of \cite{zhou}, means that $(\widehat{A}_r-sI,\widehat{B}_r)$ is stabilizable and $(C_r,\widehat{A}_r-sI)$ is detectable. Note that $\footnotesize\begin{bmatrix}
	\mathbf{N}^\top(s)&\mathbf{M}^\top(s)
	\end{bmatrix}^\top=C_r(sI-\widehat{A}_r)^{-1}\widehat{B}_r+D_r$ has full column rank $\forall s\in j{\mathbb{R}}\cup\{\infty\}$, or else there cannot exist $\widetilde{\mathbf{X}}, \widetilde{\mathbf{Y}}$ stable so that $\widetilde{\mathbf{Y}}\mathbf{M}-\widetilde{\mathbf{X}}\mathbf{N}=I$. Then, by point $(a)$ in Corollary 13.23 of \cite{zhou}, \eqref{eq:CTARE} has a stabilizing solution, $\widehat{X}_r$. Thus, \eqref{eq:GCARE} has a stabilizing solution, $X_r=(E_r^\top)^{-1}\widehat{X}_rE_r^{-1}$, and its stabilizing feedback equals that of \eqref{eq:CTARE}, $\widehat{F}_r:=-(D_r^{\top}D_r)^{-1}(\widehat B_r^{\top}\widehat X_r+D_r^{\top}C_r)=F_r$. \qed
	
	\textbf{Proof of Theorem \ref{thm:robstab}:} To prove point $\mathbf{(a)}$, define first
	$
	{\mathbf{T}}:=\footnotesize\left[\begin{array}{c:c}
	{\mathbf{T}}_{11}&{\mathbf{T}}_{12}\\\hdashline{\mathbf{T}}_{21}&{\mathbf{T}}_{22}
	\end{array}\right]=\left[\scriptsize\begin{array}{ll:l}
	0&-\widehat{\mathbf{M}}^{-1}&\widehat{\mathbf{M}}^{-1}\\\hdashline I&-\mathbf{G}&\mathbf{G}
	\end{array}\right],
	$
	where $\widehat{\mathbf{M}}^{-1}=\mathbf{G}_0\mathbf{M}^{-1}$ and $\mathbf{G}_0$ is expressed as in Lemma \ref{lem:NRCF}. Moreover, we have from \eqref{eq:dublu2} in Theorem \ref{thm:Youla} that
	$
	{\mathbf{M}}^{-1}=\footnotesize\left[\begin{array}{c|c}
	A-sE&B\\\hline-F&I
	\end{array}\right].
	$
	Expressing ${\mathbf{T}}=\footnotesize\left[\begin{array}{c:c}
	\mathbf{G}_0&0\\\hdashline 0&I
	\end{array}\right]\left[\begin{array}{ll:l}
	0&-{\mathbf{M}}^{-1}&{\mathbf{M}}^{-1}\\\hdashline I&-\mathbf{G}&\mathbf{G}
	\end{array}\right]$ and noticing that $\mathbf{T}^{\epsilon}=\footnotesize\begin{bmatrix}
	\epsilon I&0\\0&I
	\end{bmatrix}{\mathbf{T}}$, for an $\epsilon\in(0,1]$, we obtain the realization given in \eqref{eq:gen_rob}. Since $\mathbf{G}$ is given by a strongly stabilizable and detectable realization \eqref{eq:tfm}, then it is always possible to find $\widetilde F$ and $\widetilde H$ so that $A+B\widetilde  F-sE$ and $A+\widetilde HC-sE$ are admissible. Thus, defining $\widehat{F}:=\small\begin{bmatrix}
	0&\widetilde F
	\end{bmatrix}$ and $\widehat{H}:=\small\begin{bmatrix}
	0&\widetilde{H}^{\top}
	\end{bmatrix}\hspace{-2mm}\phantom{.}^{\top}$, we extract the realization of $\mathbf{T}_{22}^\epsilon$ from \eqref{eq:gen_rob},
	$
	\mathbf{T}_{22}^\epsilon=\footnotesize\left[\begin{array}{c|c}
	A_{22}-sE_{22}&B_{22}\\\hline C_{22}&D_{22}
	\end{array}\right]:=\tiny\left[\begin{array}{cc|cc:c}
	A_r-sE_r&\hspace{-2mm}-B_rF&B_r\\
	0&\hspace{-2mm}A-sE&B\\\hline
	0&\hspace{-2mm}C&D
	\end{array}\right],
	$
	to get that $A_{22}+B_{22}\widehat{F}-sE_{22}$ and $A_{22}+\widehat{H}C_{22}-sE_{22}$ are both admissible, since $\Lambda(A_r-sE_r)\subset\mathbb{C}^-$. Therefore, $\widehat{F}$ and $\widehat{H}$ can be used, as in Theorem \ref{thm:Youla}, in order to express the class of stabilizing controllers via a DCF over $\mathcal{RH}_\infty$ of $\mathbf{T}_{22}^\epsilon$.
	
	To prove point $\mathbf{(b)}$, begin by defining the system\vspace{-2mm}
	\begin{equation}\label{eq:T_bar}
		\hspace{-1mm}\overline{\mathbf{T}}:=\hspace{-1mm}\tiny\left[\begin{array}{cc|cc:c}
			A_r-sE_r&-B_rF&0&-B_r&B_r\\
			0&A-sE&0&-B&B\\\hline
			- H_rF_r&- H_rF&0&- H_r& H_r\\\hdashline
			0&0&0&\phantom{-}0&I\\
			0&C&I&-D&D
		\end{array}\right]\hspace{-1mm}=\hspace{-1mm}\small\left[\begin{array}{c:c}
			\overline{\mathbf{T}}_{11}&\overline{\mathbf{T}}_{12}\\\hdashline\overline{\mathbf{T}}_{21}&\overline{\mathbf{T}}_{22}
		\end{array}\right]\normalsize\hspace{-2mm}\vspace{-2mm}
	\end{equation}
	and by considering the class of TFMs expressed through
	$\mathcal{F}_u\Big(\overline{\mathbf{T}},\footnotesize\begin{bmatrix}
		\mathbf{\Delta}_{\widehat{\mathbf{N}}}\\\mathbf{\Delta}_{\widehat{\mathbf{M}}}
	\end{bmatrix}\normalsize\Big):=\overline{\mathbf{T}}_{22}+\overline{\mathbf{T}}_{21}\footnotesize\begin{bmatrix}
		\mathbf{\Delta}_{\widehat{\mathbf{N}}}\\\mathbf{\Delta}_{\widehat{\mathbf{M}}}
	\end{bmatrix}\normalsize\Big(I-\overline{\mathbf{T}}_{11}$ $\footnotesize\begin{bmatrix}
		\mathbf{\Delta}_{\widehat{\mathbf{N}}}\\\mathbf{\Delta}_{\widehat{\mathbf{M}}}
	\end{bmatrix}\normalsize\Big)^{-1}\overline{\mathbf{T}}_{12}$, with $\mathbf{\Delta}_{\widehat{\mathbf{N}}}$ and $\mathbf{\Delta}_{\widehat{\mathbf{M}}}$ as in \eqref{eq:class}. Denoting now the class of TFMs $\mathbf{G}_{\mathbf{\Delta}}:=\big( \widehat{\mathbf N}+\mathbf{\Delta}_{\widehat{\mathbf{N}}}\big)\big(\widehat{\mathbf M}+\mathbf{\Delta}_{\widehat{\mathbf{M}}}\big)^{-1}$, it is straightforward to check that $\mathcal{F}_u\Big(\overline{\mathbf{T}},\footnotesize\begin{bmatrix}
	\mathbf{\Delta}_{\widehat{\mathbf{N}}}\\\mathbf{\Delta}_{\widehat{\mathbf{M}}}
\end{bmatrix}\normalsize\Big)=\begin{bmatrix}
I&\mathbf{G}_{\mathbf{\Delta}}^\top
\end{bmatrix}\hspace{-1mm}\phantom{}^\top$. Thus, the proof of point $\mathbf{(b)}$ boils down to applying the Small Gain Theorem, as formulated in Chapter 8 of \cite{ess}, to confirm robust stability.

	Note that, since the realization of ${\mathbf{T}}^\epsilon_{22}$ from \eqref{eq:gen_rob} is strongly stabilizable and detectable, then so is the one belonging to $\overline{\mathbf{T}}_{22}^{\epsilon}$ in \eqref{eq:T_bar}. Now, if $(\mathbf{\Phi},\mathbf{\Gamma})$ is an NRF implementation of $\mathbf{K}$ as in \eqref{eq:NRF_def_1}-\eqref{eq:NRF_def_2} and $\mathbf{K}$ stabilizes $\mathbf{G}$, then $\begin{bmatrix}
		\mathbf{\Phi}&\mathbf{\Gamma}
	\end{bmatrix}$ stabilizes $\overline{\mathbf{T}}_{22}=\begin{bmatrix}
	I&\mathbf{G}^\top
\end{bmatrix}\hspace{-1mm}\phantom{}^\top$ (see \cite{NRF}). Since the latter's realization in \eqref{eq:T_bar} is strongly stabilizable and detectable, then $\begin{bmatrix}
\mathbf{\Phi}&\mathbf{\Gamma}
\end{bmatrix}$ stabilizes $\overline{\mathbf{T}}$ (as in Theorem \ref{thm:Youla}). Finally, it is straightforward to check that $\mathcal{F}_\ell({\mathbf{T}}^{\epsilon},{\mathbf{K}})=\epsilon\mathcal{F}_\ell\big(\overline{\mathbf{T}},\begin{bmatrix}
\mathbf{\Phi}&\mathbf{\Gamma}
\end{bmatrix}\big)$. If $
\left\|\mathcal{F}_\ell({\mathbf{T}}^{\epsilon},{\mathbf{K}})\right\|_\infty\hspace{-1mm}\leq1
$ then $\left\|\mathcal{F}_\ell\big(\overline{\mathbf{T}},\begin{bmatrix}
\mathbf{\Phi}&\mathbf{\Gamma}
\end{bmatrix}\big)\right\|_\infty\leq\frac{1}{\epsilon}$ and, by applying point $(b)$ of Theorem 8.1 in \cite{ess}, it follows that the closed-loop interconnection between $\begin{bmatrix}
\mathbf{\Phi}&\mathbf{\Gamma}
\end{bmatrix}$ and $\mathcal{F}_u\Big(\overline{\mathbf{T}},\footnotesize\begin{bmatrix}
\mathbf{\Delta}_{\widehat{\mathbf{N}}}\\\mathbf{\Delta}_{\widehat{\mathbf{M}}}
\end{bmatrix}\normalsize\Big)=\begin{bmatrix}
I&\mathbf{G}_{\mathbf{\Delta}}^\top
\end{bmatrix}\hspace{-1mm}\phantom{}^\top$ will be internally stable and well-posed for any $\mathbf{G}_{\mathbf{\Delta}}\in\mathcal{C}_{\mathbf{G}}^\epsilon$.

\noindent
As shown in the proof of the main result from \cite{NRF}, this ensures that the control laws from \eqref{eq:NRF_explicit} will stabilize any $\mathbf{G}_{\mathbf{\Delta}}\in\mathcal{C}_{\mathbf{G}}^\epsilon$. \qed

	\textbf{Proof of Proposition \ref{prop:size_reduce}:} Define first $\widehat{F}:=\begin{bmatrix}
	0&F
	\end{bmatrix}$ and $\widehat{H}:=\begin{bmatrix}
	0&H^{\top}
	\end{bmatrix}\hspace{-1mm}\phantom{}^{\top}$ and employ these two feedbacks to write, via \eqref{eq:dublu1}-\eqref{eq:dublu2}, a DCF of $\mathbf{T}_{22}^\epsilon$ from \eqref{eq:gen_rob}. This factorization is indeed a DCF over $\mathcal{RH}_\infty$ due to the fact that $A+BF-sE$ and $A+HC-sE$ are admissible and $\Lambda(A_r-sE_r)\subset\mathbb{C}^-$. The identities from \eqref{eq:dublu3}-\eqref{eq:dublu4} and \eqref{eq:afin5}-\eqref{eq:afin7} follow by writing the realizations given by \eqref{eq:dublu1}-\eqref{eq:dublu2} and by \eqref{eq:afin2}-\eqref{eq:afin4} in Theorem \ref{thm:Youla}, and then eliminating all unobservable modes.\qed
	
	\textbf{Proof of Theorem \ref{thm:algo}}: To prove point $\mathbf{(a)}$, we first ensure that the realization from \eqref{eq:Q_hat_ss} is stable by imposing that $A_x^\top$ and $\left(A_i^\mathbf{B}+B_i^\mathbf{B}F_i^\mathbf{B}\right)\hspace{-2mm}\phantom{.}^\top,\ \forall\ i\in 1:q$, have eigenvalues only in $\mathbb{C}^-$, along with $A_{\overline{\mathbf{B}}}^\top$ via \eqref{eq:B_fac}-\eqref{eq:B_bar}. These conditions are equivalent to $\exists\  {P}_x={P}_x^\top\succ 0$ and $P_i^\mathbf{B}=\left(P_i^\mathbf{B}\right)\hspace{-2mm}\phantom{.}^\top\succ 0$, $i\in 1:q$, such that $-2\text{sym}(A_xP_x)\hspace{-0.5mm}\succ\hspace{-0.5mm} 0$ and $-2\text{sym}\left(A_i^\mathbf{B}P_i^\mathbf{B}+B_i^\mathbf{B}F_i^\mathbf{B}P_i^\mathbf{B}\right)\hspace{-0.5mm}\succ\hspace{-0.5mm} 0,\ \forall i\in1: q$. To remedy the bilinearity induced by $P_xA_x^\top$ and $P_i^\mathbf{B}\left(F_i^\mathbf{B}\right)\hspace{-2mm}\phantom{.}^\top$, define $\overline{P}_x:=P_xA_x^\top$ and $\overline{P}_i^\mathbf{B}:=P_i^\mathbf{B}(F_i^\mathbf{B})\hspace{-1mm}\phantom{.}^\top$, with $i\in 1:q$, and rewrite the inequalities as $-2\text{sym}(\overline{P}_x)\succ 0$ and $-2\text{sym}\big(P_i^\mathbf{B}\left(A_i^\mathbf{B}\right)\hspace{-2mm}\phantom{.}^\top+\overline{P}_i^\mathbf{B}\left(B_i^\mathbf{B}\right)\hspace{-2mm}\phantom{.}^\top\big)\succ 0,\ \forall i\in1: q$.

	If these new affine inequalities are satisfied, then due to $\Lambda\left(A^f-sI\right)\subset\mathbb{C}^-$ and to $C_2^f\big(sI-A^f\big)^{-1}B_2^f\equiv 0$, it follows that $\overline{A}$ from \eqref{eq:H_CL_ss} has $\Lambda\left(\overline{A}-sI\right)\subset\mathbb{C}^-$. By the equivalence of points $(i)$ and $(vii)$ from Corollary 12.3 in \cite{ess}, we have that $\left\|{\mathbf{T}}_1^\epsilon+{\mathbf{T}}_2^\epsilon\mathbf{Q}{\mathbf{T}}_3^\epsilon\right\|_\infty<1$ if and only if $\exists\ P=P^{\top}\succ 0$ such that
	$\tiny
	-\begin{bmatrix}
	2\text{sym}(P\overline{A})&P\overline{B}&\phantom{-}\overline{C}^{\top}\\
	\overline{B}^{\top}P&- I&\phantom{-}\overline{D}^{\top}\\
	\overline{C}&\overline{D}&- I
	\end{bmatrix}\succ 0
	$, which contains bilinear products of $P$ with $\overline{A}$ and $\overline{B}$, thus leading to nonconvex optimization. To obtain an affine expression, define $A_S$ as in \eqref{eq:aux3}, in order to introduce $\overline{P}:=PA_S$. With this new matrix and the four matrices defined in \eqref{eq:aux1}, notice that $P\overline{A}= P\overline{A}^f+\overline{P}\overline{C}_2^f$ and that $P\overline{B}= P\overline{B}_1^f+\overline{P}\overline{D}_{21}^f$. The norm condition is equivalent to $-G\succ 0$, with $G$ from \eqref{eq:aux2} being affine in all variables.
	
	Recall the inequalities from \eqref{eq:prob_x}, that contain bilinear terms, to denote $\overline{P}^D:=\overline{d}_xd_x^{\top}$, while imposing that $\overline{d}_x-d_x=0$. The latter will also induce the additional constraint $\overline{P}^D = \big(\overline{P}^D\big)^\top$, from which we obtain the $N_Z$ LMIs given in \eqref{eq:prob_iter}. Form now the matrices from \eqref{eq:aux3}-\eqref{eq:aux6} to note that \eqref{eq:prob_x} is equivalent to\vspace{-1mm}
	\begin{equation}\label{eq:prob_BE}\footnotesize
	\left\{\hspace{-2mm}
	\begin{array}{l}\vspace{-1mm}
	(\widehat{Z}_1^k)^\top\widehat{Z}_1^k+\hspace{-1mm}\sum\limits_{i\in1:q}\hspace{-1mm}d_{xi}\big((\widehat{Z}_1^k)^\top\widehat{Z}_{2i}^k+(\widehat{Z}_{2i}^k)^\top\widehat{Z}_1^k\big)+\hspace{-1mm}\sum\limits_{i\in1:q}\hspace{-1mm}\overline{P}^D_{ii}(\widehat{Z}_{2i}^k)^\top\widehat{Z}_{2i}^k+\\
	+\sum_{i=1}^{q-1}\sum_{j=i+1}^{q}\overline{P}^D_{ij}\big((\widehat{Z}_{2i}^k)^\top\widehat{Z}_{2j}^k+(\widehat{Z}_{2j}^k)^\top\widehat{Z}_{2i}^k\big)\succ 0,\forall k\in1:N_Z,\\
	\overline{P}^D = \big(\overline{P}^D\big)^\top,\ P=P^{\top}\succ 0,\ -G\succ 0,\ P_x=P_x^\top\succ 0,\\
	-2\text{sym}(\overline{P}_x)\succ 0,\ P_i^\mathbf{B}=(P_i^\mathbf{B})^{\top}\succ 0,\ \forall i\in 1:q,\\
	-2\text{sym}\left(P_i^\mathbf{B}\left(A_i^\mathbf{B}\right)^\top+\overline{P}_i^\mathbf{B}\left(B_i^\mathbf{B}\right)^\top\right)\succ 0,\ \forall i\in 1:q,\ T_AT_B=T_C.\vspace{-1mm}
	\end{array}
	\right.\normalsize\vspace{-1mm}
	\end{equation}
	
	By selecting an artificial scalar $\gamma>0$ as the cost function and by applying Theorem 1 in \cite{BMI2LMI}, we get that \eqref{eq:prob_BE} is equivalent to the problem in which $T_AI_{n_T}T_B=T_C$ is replaced by $\text{rank}\footnotesize\begin{bmatrix}
		T_C+XY-T_AY-XT_B&T_A-X\\T_B-Y&I_{n_T}
	\end{bmatrix}\normalsize\hspace{-1mm}=\text{rank}\ I_{n_T}$, for any matrices $X$ and $Y$. Therefore, by applying Theorem 2 in \cite{BMI2LMI} with a regularization parameter $\lambda>0$, adapting Algorithm 1 in \cite{BMI2LMI} for the resulting problem, scaling its cost function by $1/\lambda$ and then taking $\lambda\rightarrow\infty$, we obtain Algorithm \ref{alg:iter}, which solves \eqref{eq:prob_iter} at each iteration. If the initialization is successful (the LMI system along with $\overline{d}_x-d_x=0$ must be feasible for the original BMI system to be feasible), we then set $\Theta^1=T_B-T_B^0$ and we employ Theorem 3 in \cite{BMI2LMI} for our algorithm (with the adapted cost function), which guarantees its convergence.
	\qed


	

	
	
	
	\bibliographystyle{IEEEtran}
	\bibliography{manuscript.bib}

\begin{thebibliography}{10}
\providecommand{\url}[1]{#1}
\csname url@samestyle\endcsname
\providecommand{\newblock}{\relax}
\providecommand{\bibinfo}[2]{#2}
\providecommand{\BIBentrySTDinterwordspacing}{\spaceskip=0pt\relax}
\providecommand{\BIBentryALTinterwordstretchfactor}{4}
\providecommand{\BIBentryALTinterwordspacing}{\spaceskip=\fontdimen2\font plus
\BIBentryALTinterwordstretchfactor\fontdimen3\font minus
  \fontdimen4\font\relax}
\providecommand{\BIBforeignlanguage}[2]{{%
\expandafter\ifx\csname l@#1\endcsname\relax
\typeout{** WARNING: IEEEtran.bst: No hyphenation pattern has been}%
\typeout{** loaded for the language `#1'. Using the pattern for}%
\typeout{** the default language instead.}%
\else
\language=\csname l@#1\endcsname
\fi
#2}}
\providecommand{\BIBdecl}{\relax}
\BIBdecl

\bibitem{Lavaei}
G.~Fazelnia, R.~Madani, A.~Kalbat, and J.~Lavaei, ``{Convex Relaxation for
  Optimal Distributed Control Problems},'' \emph{IEEE Trans. on Automatic
  Control}, vol.~62, no.~1, pp. 206--221, 2017.

\bibitem{reg}
N.~Matni and V.~Chandrasekaran, ``{Regularization for Design},'' \emph{IEEE
  Trans. on Automatic Control}, vol.~61, no.~12, pp. 3991--4006, 2016.

\bibitem{Sznaier}
Y.~Wang, J.~A. Lopez, and M.~Sznaier, ``{Convex Optimization Approaches to
  Information Structured Decentralized Control},'' \emph{IEEE Trans. on
  Automatic Control}, vol.~63, no.~10, pp. 3393--3403, 2018.

\bibitem{trian}
C.~A. Rösinger and C.~W. Scherer, ``{A Flexible Synthesis Framework of
  Structured Controllers for Networked Systems},'' \emph{IEEE Trans. on Control
  of Network Systems}, vol.~7, no.~1, pp. 6--18, 2020.

\bibitem{Alav}
A.~Alavian and M.~C. Rotkowitz, ``{Q-Parametrization and an SDP for
  $\mathcal{H}_\infty$-optimal Decentralized Control},'' \emph{IFAC Proceedings
  Volumes}, vol.~46, no.~27, pp. 301--308, 2013.

\bibitem{QI_orig}
M.~{Rotkowitz} and S.~{Lall}, ``{A Characterization of Convex Problems in
  Decentralized Control},'' \emph{IEEE Trans. on Automatic Control}, vol.~51,
  no.~2, pp. 274--286, 2006.

\bibitem{SLA}
Y.~{Wang}, N.~{Matni}, and J.~C. {Doyle}, ``{A System-Level Approach to
  Controller Synthesis},'' \emph{IEEE Trans. on Automatic Control}, vol.~64,
  no.~10, pp. 4079--4093, 2019.

\bibitem{plutonizare}
\c{S}. {Sab\u{a}u}, C.~{Oar\u{a}}, S.~{Warnick}, and A.~{Jadbabaie}, ``{Optimal
  Distributed Control for Platooning via Sparse Coprime Factorizations},''
  \emph{IEEE Trans. on Automatic Control}, vol.~62, no.~1, pp. 305--320, 2017.

\bibitem{NRF}
\c{S}. Sab\u{a}u, A.~{Speril\u{a}}, C.~Oar\u{a}, and A.~Jadbabaie, ``{Network
  Realization Functions for Optimal Distributed Control},'' \emph{IEEE Trans.
  on Automatic Control, to be published}, 2024.

\bibitem{VSOLVE}
A.~Varga, ``Computation of least order solutions of linear rational
  equations,'' \emph{{In Proc. of the International Symposium on Mathematical
  Theory of Networks and Systems}}, 2004.

\bibitem{VNULL}
------, ``{On Computing Nullspace Bases - a Fault Detection Perspective},''
  \emph{In Proc. of the $17^{th}$ IFAC World Congress}, pp. 6295--6300, 2008.

\bibitem{VCOVER}
------, ``Reliable algorithms for computing minimal dynamic covers for
  descriptor systems,'' \emph{{In Proc. of the International Symposium on
  Mathematical Theory of Networks and Systems}}, 2004.

\bibitem{glover1989robust}
K.~Glover and D.~McFarlane, ``{Robust stabilization of normalized coprime
  factor plant descriptions with $\mathcal{H}_\infty$-bounded uncertainty},''
  \emph{IEEE Trans. on Automatic Control}, vol.~34, no.~8, pp. 821--830, 1989.

\bibitem{BMI2LMI}
R.~{Doelman} and M.~{Verhaegen}, ``{Sequential convex relaxation for convex
  optimization with bilinear matrix equalities},'' \emph{In Proc. of the 2016
  European Control Conference}, pp. 1946--1951, 2016.

\bibitem{gantmacher}
F.~Gantmacher, \emph{{The Theory of Matrices}}.\hskip 1em plus 0.5em minus
  0.4em\relax American Math. Society, 1959.

\bibitem{Kai}
T.~Kailath, \emph{{Linear Systems}}.\hskip 1em plus 0.5em minus 0.4em\relax
  Prentice-Hall, 1980.

\bibitem{COAV}
C.~Oar\u{a} and A.~Varga, ``{Minimal Degree Coprime Factorization of Rational
  Matrices},'' \emph{SIAM Journal on Matrix Analysis and Applications},
  vol.~21, no.~1, pp. 245--278, 1999.

\bibitem{VPLACE}
A.~Varga, ``On stabilization methods of descriptor systems,'' \emph{Systems \&
  Control Letters}, vol.~24, no.~2, pp. 133--138, 1995.

\bibitem{RiccBig}
P.~Benner, Z.~Bujanović, P.~Kürschner, and J.~Saak, ``{A Numerical Comparison
  of Different Solvers for Large-Scale, Continuous-Time Algebraic Riccati
  Equations and LQR Problems},'' \emph{SIAM Journal on Scientific Computing},
  vol.~42, pp. 957--996, 2020.

\bibitem{takaba}
K.~Takaba, N.~Morihira, and T.~Katayama, ``{$\mathcal{H}_\infty$ control for
  descriptor systems: A J-spectral factorization approach},'' \emph{In Proc. of
  $33^{rd}$ IEEE Conference on Decision and Control}, pp. 2251--2256, 1994.

\bibitem{VNCF}
A.~Varga, ``Computation of normalized coprime factorizations of rational
  matrices,'' \emph{Systems \& Control Letters}, vol.~33, no.~1, pp. 37--45,
  1998.

\bibitem{gap}
T.~Georgiou and M.~Smith, ``Optimal robustness in the gap metric,'' \emph{IEEE
  Trans. on Automatic Control}, vol.~35, no.~6, pp. 673--686, 1990.

\bibitem{zhou}
K.~Zhou, J.~C. Doyle, and K.~Glover, \emph{{Robust and Optimal Control}}.\hskip
  1em plus 0.5em minus 0.4em\relax Prentice-Hall, 1996.

\bibitem{BF}
B.~A. Francis, \emph{{A Course in $\mathcal{H}_\infty$ Control Theory}}.\hskip
  1em plus 0.5em minus 0.4em\relax Springer-Verlag, 1987.

\bibitem{ADMM}
C.~Sun and R.~Dai, ``{A customized ADMM for rank-constrained optimization
  problems with approximate formulations},'' \emph{In Proc. of the
  $56^\text{th}$ IEEE Conference on Decision and Control}, pp. 3769--3774,
  2017.

\bibitem{IRM}
------, ``Rank-constrained optimization and its applications,''
  \emph{Automatica}, vol.~82, pp. 128--136, 2017.

\bibitem{mosek}
{MOSEK ApS}, \emph{{The MOSEK optimization toolbox for MATLAB manual. Version
  10.0}}, 2022.

\bibitem{YALMIP}
J.~L{\"{o}}fberg, ``{YALMIP : A Toolbox for Modeling and Optimization in
  MATLAB},'' \emph{In Proc. of the CACSD Conference}, 2004.

\bibitem{rank}
M.~Fazel, H.~Hindi, and S.~Boyd, ``{A Rank Minimization Heuristic with
  Application to Minimum Order System Approximation},'' \emph{In Proc. of the
  2001 American Control Conference}, vol.~6, pp. 4734--4739, 2001.

\bibitem{ess}
K.~Zhou and J.~C. Doyle, \emph{{Essentials of Robust Control}}.\hskip 1em plus
  0.5em minus 0.4em\relax Prentice-Hall, 1998.

\end{thebibliography}

\end{document}